\documentclass[11pt]{article}

\usepackage[T1]{fontenc}
\usepackage[utf8]{inputenc}
\usepackage[english]{babel}
\usepackage{amsmath,amssymb,amsthm}
\usepackage{thmtools,thm-restate}
\usepackage{mathtools}
\mathtoolsset{centercolon}
\usepackage[noend]{algpseudocode}
\usepackage{algorithm}
\usepackage{fullpage}
\usepackage{enumitem}
\usepackage{url}
\usepackage{graphicx}
\usepackage[stable]{footmisc}
\usepackage{comment}
\usepackage[textsize=scriptsize]{todonotes}
\usepackage[margin=1.5cm]{caption}
\usepackage{subcaption}
\usepackage{tikz}
\usetikzlibrary{positioning,plotmarks,calc,arrows,decorations.markings,decorations.pathreplacing,intersections,fit,matrix,shapes}
\usepackage[hidelinks]{hyperref}

\declaretheorem{theorem}
\declaretheorem[sibling=theorem]{definition}
\declaretheorem[sibling=theorem]{lemma}
\declaretheorem[sibling=theorem]{corollary}
\declaretheorem[sibling=theorem]{fact}

\newcommand\numberthis{\addtocounter{equation}{1}\tag{\theequation}}

\newcommand{\OPT}{\ensuremath{\mathit{OPT}}}
\newcommand{\poly}{\ensuremath{\text{poly}}}
\newcommand{\EMD}{\ensuremath{\textnormal{EMD}}}

\newcommand{\R}{\ensuremath{\mathbb{R}}}
\newcommand{\reals}{\ensuremath{\mathbb{R}}}
\newcommand{\N}{\ensuremath{\mathbb{N}}}
\newcommand{\eps}{\varepsilon}
\newcommand{\xhat}{\ensuremath{\widehat{x}}}
\newcommand{\pinv}{^\dag}
\newcommand{\med}[1]{\textnormal{MED}(#1)}
\newcommand{\Asupports}{\ensuremath{\mathbb{A}}}
\newcommand{\Bsupports}{\ensuremath{\mathbb{B}}}
\newcommand{\Csupports}{\ensuremath{\mathbb{C}}}
\newcommand{\Msupports}{\ensuremath{\mathbb{M}}}
\newcommand{\Msupportsclosure}{\ensuremath{\mathbb{M}^+}}
\newcommand{\Mmodel}{\ensuremath{\mathcal{M}}}
\newcommand{\Amodel}{\ensuremath{\mathcal{A}}}
\newcommand{\Bmodel}{\ensuremath{\mathcal{B}}}
\newcommand{\Cmodel}{\ensuremath{\mathcal{C}}}
\newcommand{\powerset}{\ensuremath{\mathcal{P}}}
\newcommand{\supp}{\ensuremath{\textnormal{supp}}}
\newcommand{\colsupp}{\ensuremath{\textnormal{col-supp}}}
\newcommand{\supportplus}{\ensuremath{\mathbin{\oplus}}}
\DeclarePairedDelimiter{\ceil}{\lceil}{\rceil}
\DeclarePairedDelimiter{\floor}{\lfloor}{\rfloor}
\DeclarePairedDelimiter{\norm}{\lVert}{\rVert}
\DeclarePairedDelimiter{\parens}{\lparen}{\rparen}
\DeclarePairedDelimiter{\abs}{\lvert}{\rvert}

\makeatletter
\let\oldparens\parens
\def\parens{\@ifstar{\oldparens}{\oldparens*}}
\let\oldnorm\norm
\def\norm{\@ifstar{\oldnorm}{\oldnorm*}}
\let\oldceil\ceil
\def\ceil{\@ifstar{\oldceil}{\oldceil*}}
\let\oldfloor\floor
\def\floor{\@ifstar{\oldfloor}{\oldfloor*}}
\let\oldabs\abs
\def\abs{\@ifstar{\oldabs}{\oldabs*}}
\makeatother

\makeatletter
\let\OldStatex\Statex
\renewcommand{\Statex}[1][3]{%
\setlength\@tempdima{\algorithmicindent}%
\OldStatex\hskip\dimexpr#1\@tempdima\relax}
\makeatother

\makeatletter
\let\oldr@@t\r@@t
\def\r@@t#1#2{%
\setbox0=\hbox{$\oldr@@t#1{#2\,}$}\dimen0=\ht0
\advance\dimen0-0.2\ht0
\setbox2=\hbox{\vrule height\ht0 depth -\dimen0}%
{\box0\lower0.4pt\box2}}
\LetLtxMacro{\oldsqrt}{\sqrt}
\renewcommand*{\sqrt}[2][\ ]{\oldsqrt[#1]{#2}}
\makeatother

\begin{document}

\title{Approximation Algorithms for \\ Model-Based Compressive Sensing}

\author{Chinmay Hegde, Piotr Indyk, Ludwig Schmidt \thanks{Authors
    listed in alphabetical order. The authors
    would like to thank Lei Hamilton, Chris Yu, Ligang Lu, and Detlef
    Hohl for helpful discussions. This work was supported in part by
    grants from the MITEI-Shell program, the MADALGO center, and the
    Packard Foundation. A conference version of this manuscript~\cite{approxSODA}
    appeared in the Proceedings of the ACM-SIAM Symposium on Discrete
    Algorithms (SODA), held in January 2014.}
\\[.2cm]CSAIL, MIT
}

\maketitle
\thispagestyle{empty}
\begin{abstract}

Compressive Sensing (CS) states that a sparse signal can be
recovered from a small number of linear measurements, and that
this recovery can be performed efficiently in polynomial time. The
framework of \emph{model-based compressive sensing} (model-CS) leverages
additional structure in the signal and provides new recovery schemes that can
reduce the number of measurements even further. This idea has led to
measurement-efficient recovery schemes for a variety of signal
models.  However, for any given model, model-CS
requires an algorithm that solves the
\emph{model-projection problem}: given a
query signal, report the signal in the model that is closest to the query signal.
Often, this optimization problem can be computationally very expensive.
Moreover, an \emph{approximation} algorithm is not sufficient to provably recover the signal. As a result, the model-projection problem poses a fundamental obstacle for extending model-CS to many interesting classes of models.  

In this paper, we introduce a new framework that we call
\emph{approximation-tolerant model-based compressive sensing}.
This framework includes a range of algorithms for sparse recovery that
require only \emph{approximate solutions} for the model-projection problem.
In essence, our work removes the aforementioned obstacle to model-based
compressive sensing, thereby extending model-CS to a much wider class of signal models.
Interestingly, all our algorithms involve both the minimization and maximization variants of the model-projection problem. 

We instantiate our new framework for a new signal model that we call
the Constrained Earth Mover Distance (CEMD) model. This model is
particularly useful for signal ensembles where the positions of the
nonzero coefficients do not change significantly as a function of
spatial (or temporal) location. We develop novel approximation algorithms
for both the maximization and the minimization versions of the
model-projection problem via graph optimization techniques. Leveraging
these algorithms and our framework results in a nearly sample-optimal
sparse recovery scheme for the CEMD model.

\end{abstract}

\section{Introduction}
\label{sec:intro}
Over the last decade, a new \emph{linear} approach for obtaining a succinct
representation of $n$-dimensional vectors (or signals) has emerged.
For any signal $x$, the representation is given by $Ax$, where $A$ is an $m \times n$ matrix, or possibly a random variable chosen from a suitable distribution over such matrices.
The vector $Ax$ is referred to as the {\em measurement vector} or \emph{linear sketch} of $x$.
Although $m$ is usually chosen to be much smaller than $n$, the measurement vector $Ax$ often contains plenty of useful information about the signal $x$.

A particularly useful and well-studied problem in this context is that of \emph{robust sparse recovery}.
A vector $x$ is $k$-sparse if it has at most $k$ non-zero coordinates.
The robust sparse recovery problem is typically defined as follows: 
given the measurement vector $y=Ax+e$, where $x$ is a $k$-sparse vector and  $e$ is the ``noise'' vector
%
, find a signal estimate $\xhat$ such that:
\begin{equation}
\label{e:lplq}
\norm{x-\xhat}_2 \; \le \; C \cdot \norm{e}_2 \; .
\end{equation}
Sparse recovery has a tremendous number of applications in areas such as compressive sensing of signals~\cite{cs_incomplete_frequency,cs_donoho}, genetic data analysis~\cite{poolmc}, and data stream algorithms~\cite{M05,cs_sparse_matrices_survey}.

It is known that there exist matrices $A$ and associated recovery algorithms that produce a signal estimate $\xhat$ satisfying Equation~\eqref{e:lplq} with a constant approximation factor $C$ and number of measurements $m=O(k \log (n/k))$.
It is also known that this bound on the number of measurements $m$ is asymptotically \emph{optimal} for some constant $C$; see~\cite{DIPW10} and~\cite{FPRU10} (building upon the classical results of~\cite{K77,GG84,G84}).  
The necessity of the ``extra'' logarithmic factor multiplying $k$ is rather unfortunate: the quantity $m$ determines the ``compression rate'', and for large $n$ any logarithmic factor in $m$ can worsen this rate tenfold.

On the other hand, more careful signal \emph{modeling} offers a way to overcome the aforementioned limitation.
Indeed, decades of research in signal processing have shown that not all signal supports (i.e., sets of non-zero coordinates) are equally common in practice.
For example, in the case of certain time-domain signals such as signals transmitted by push-to-talk radios, the dominant coefficients of the signal tend to cluster together in contiguous ``bursts''.
A formal approach to capture this additional \emph{structure} is to assume that the support of the vector $x$ belongs to a given family of supports $\Msupports$, a so-called ``model'' (we say that $x$ is  $\Msupports$-{sparse}).
Note that the original $k$-sparse recovery problem corresponds to the particular case when the model $\Msupports$ is the family of all $k$-subsets of $[n]$.

This modeling approach has several interesting ramifications, particularly in the context of robust sparse recovery.
Recently, Baraniuk et al.\ provided a general framework called \emph{model-based compressive sensing} \cite{modelcs}.
For any ``computationally tractable'' and ``small'' family of supports, the scheme proposed in their work guarantees robust signal recovery with a nearly-optimal number of measurements $m=O(k)$, i.e.,  \emph{without any logarithmic dependence on $n$.}
Several other works have achieved similar performance gains both in theory and in practice; see, for example,~\cite{EldarUSS,DE11,rao2012universal,BJMO12a,Wain14}.

While the model-based compressive sensing framework is general, it relies on two model-specific assumptions: 
\begin{enumerate}[label={(\arabic*)}]
\item\emph{Model-based Restricted Isometry Property (RIP)}: The matrix $A$ approximately preserves the $\ell_2$-norm of all $\Msupports$-sparse vectors. 
\item\emph{Model projection oracle}: There exists an efficient algorithm that solves the {\em model-projection problem}: given an arbitrary vector $x$, the algorithm finds the $\Msupports$-sparse vector $x'$ that is closest to $x$, i.e., minimizes the $\ell_2$-norm of the ``tail'' error $\|x-x'\|_2$.
\end{enumerate}
By constructing matrices satisfying (1) and algorithms satisfying (2), researchers have developed robust signal recovery schemes for a wide variety of signal models, including block-sparsity~\cite{modelcs}, tree-sparsity~\cite{modelcs}, clustered sparsity~\cite{ModelCSSAMPTA}, and separated spikes~\cite{spikes}, to name a few.

Unfortunately, extending the model-based compressive sensing framework to more general models faces a significant obstacle.
For the framework to apply, the model projection oracle has to be \emph{exact} (i.e., the oracle finds the signal in the model with exactly minimal tail error).
This fact may appear surprising, but in Section~\ref{sec:counterexample} we provide a \emph{negative result} and prove that existing model-based recovery approaches fail to achieve the robust sparse recovery criterion \eqref{e:lplq} if the model projection oracle is not exact.
Consequently, this burden of ``exactness'' excludes several useful design paradigms employed in {\em approximation algorithms}, i.e., algorithms which find a signal in the model that only approximately minimizes the tail error.
A rich and extensive literature on approximation algorithms has emerged over the last 15 years, encompassing a variety of techniques such as greedy optimization, linear programming (LP) rounding, semidefinite programming (SDP) rounding, and Lagrangian relaxation. To the best of our knowledge, existing approaches for the model projection problem have largely focused on exact optimization techniques (e.g. dynamic programming~\cite{modelcs,CT13,ModelCSSAMPTA}, solving LPs without an integrality gap~\cite{spikes}, etc.).

\subsection{Summary of Our Results}

In this paper, we introduce a new framework that we call \emph{approximation-tolerant model-based compressive sensing}.
This framework includes a range of algorithms for sparse recovery that require only \emph{approximate solutions} for the model-projection problem. In essence, our work removes the aforementioned obstacle to model-based compressive sensing and therefore extends the framework to a much wider class of models.
Simultaneously, our framework provides a principled approach to leverage the wealth of approximation algorithms for recovering structured sparse signals from linear measurements.

Instead of requiring one exact model projection oracle, our algorithms assume the existence of \emph{two} oracles with complementary approximation guarantees: (i) Given $x \in \R^n$, a \emph{tail approximation} oracle returns a support $\Omega_t$ in the model such that the norm of the tail $\norm{x-x_{\Omega_t}}_2$ is approximately minimized. (ii) A \emph{head approximation} oracle  returns a support $\Omega_h$ in the model such that the norm of the head $\norm{x_{\Omega_h}}_2$ is approximately maximized.
Formally, we have:
\begin{align}
\norm{x - x_{\Omega_t}}_2 \; &\leq \; c_T \cdot \min_{\Omega \in \Msupports} \norm{x - x_\Omega}_2 \qquad\textrm{and} \label{eq:tail} \\
\norm{x_{\Omega_h}}_2 \; &\geq \; c_H \cdot \min_{\Omega \in \Msupports} \norm{x_\Omega}_2 \label{eq:head}
\end{align}
for some positive constants $c_H \leq 1$ and $c_T \geq 1$. Given access to these approximation oracles, we prove the following main result.

\begin{theorem}[Signal recovery]
\label{thm:mainresult}
Consider a structured sparsity model $\Mmodel \subseteq \reals^n$ and norm parameter $p \in \{1,2\}$.
Suppose that $x \in \Mmodel$ and that we observe $m$ noisy linear measurements $y = Ax + e$. 
Suppose further that $A$ satisfies the model-RIP in terms of the $\ell_p$-norm, and that we are given access to head- and tail-approximation oracles $H(\cdot)$ and $T(\cdot)$ satisfying \eqref{eq:tail} and \eqref{eq:head}, respectively. Then there exists an efficient algorithm that outputs a signal estimate $\xhat$ such that $\norm{x - \xhat}_p \leq C \norm{e}_p $ for some constant $C > 0$. 
\end{theorem}

We analyze the two cases $p=1$ and $p=2$ separately and develop two different types of recovery algorithms.
The case of $p=2$ is perhaps more well-studied in the literature and corresponds to the ``standard'' notion of the RIP.
In this case, our recovery algorithms are extensions of IHT, CoSaMP, and their model-based counterparts \cite{cosamp,iht,modelcs}.
The case of $p=1$ has received attention in recent years and is applicable to the situation where the measurement matrix $A$ is \emph{itself} sparse.
In this case, our recovery algorithms are extensions of those developed in~\cite{BGIK09,RF13,BBC14}. 
For both types of algorithms, the sequence of signal estimates $(x_k)$ produced by our algorithms exhibits \emph{geometric convergence} to the true signal $x$, i.e., the norm of the error $\norm{x_k-x}_p$ decreases by at least a constant factor in every iteration.
The rate of convergence depends on the approximation constants $c_T$ and $c_H$, as well as the RIP constants of the matrix $A$. 

As a case study, we instantiate both the $p=1$ and $p=2$ cases in the context of the {\em Constrained Earth Mover's Distance} (CEMD) model introduced in~\cite{HIS}.
In this model, the signal coefficients form an $h \times w$ grid and the support of each column has size at most $s$, for  $n=h \cdot w$ and $k=s \cdot w$.
For each pair of consecutive columns, say $c$ and $c'$, we define the Earth Mover's Distance (EMD) between them to be the minimum cost of matching the support sets of $c$ and $c'$ when viewed as point sets on a line.
A signal support is said to belong to the CEMD model with ``budget'' $B$ if the sum of all EMD distances between the consecutive columns is at most $B$.
See Section~\ref{s:emd} for a formal definition.
Our framework leads to the first \emph{nearly sample-optimal} recovery scheme for signals belonging to this model.
The result is obtained by designing a novel head-approximation algorithm and proving approximation guarantees for the tail-approximation algorithm that was first described in~\cite{HIS}.

\subsection{Paper Outline}

This paper includes the following contributions, organized by section.
Before our contributions, we briefly review some background in Section \ref{sec:prelims}.

\paragraph{A negative result for approximation oracles.}
In Section~\ref{sec:counterexample}, we begin with the following negative result:
combining an apprixmate model-projection oracle with the existing model-based compressive sensing approach of Baraniuk et al.\ \cite{modelcs} \emph{does not suffice} to guarantee robust signal recovery for even the most trivial model.
This serves as the motivation for a more sophisticated approach, which we develop throughout the rest of the paper.

\paragraph{Approximate model-iterative hard threholding (AM-IHT).} 
In Section~\ref{sec:amiht}, we propose a new extension of the iterative hard thresholding (IHT) algorithm~\cite{iht}, which we call \emph{approximate model iterative hard thresholding} (or AM-IHT).
Informally, given head- and tail-approximation oracles and measurements $y = Ax + e$ with a matrix $A$ satisfying the model-RIP, AM-IHT returns a signal estimate $\widehat{x}$ satisfying \eqref{e:lplq}.
We show that AM-IHT exhibits {geometric convergence}, and that the recovery guarantee for AM-IHT is asymptotically equivalent to the best available guarantees for model-based sparse recovery, despite using only approximate oracles.

\paragraph{Approximate model-CoSaMP (AM-CoSaMP).} 
In Section~\ref{sec:amcosamp}, we propose a new extension of the compressive sampling matching pursuit algorithm (CoSaMP)~\cite{cosamp}, which we call {\em approximate model CoSaMP} (or AM-CoSAMP).
As with AM-IHT, our proposed AM-CoSaMP algorithm requires a head-approximation oracle and a tail-approximation oracle.
We show that AM-CoSaMP also exhibits geometric convergence, and that the recovery guarantee for AM-CoSaMP, as well as the RIP condition on $A$ required for successful signal recovery, match the corresponding parameters for AM-IHT up to constant factors.

\paragraph{AM-IHT with sparse measurement matrices.}
In Section~\ref{sec:modelrip1}, we show that an approximation-tolerant approach similar to AM-IHT succeeds even when the measurement matrix $A$ is {\em itself} sparse.
Our approach leverages the notion of the restricted isometry property in the $\ell_1$-norm, also called the {\em RIP-1}, which was first introduced in~\cite{BGIK09} and developed further in the model-based context by~\cite{IR13,RF13,BBC14}.  
For sparse $A$, we propose a modification of AM-IHT, which we call \emph{AM-IHT with RIP-1}.
Our proposed algorithm also exhibits geometric convergence under the {\em model RIP-1} assumption on the measurement matrix $A$.

\paragraph{Compressive sensing with the CEMD Model.}
We design both head- and tail-approximation algorithms for the CEMD model:
(i) Our tail-approximation oracle returns a support set with tail-approximation error at most a constant times larger than the optimal tail error.
At the same time, the EMD-budget of the solution is still $O(B)$ (Theorem~\ref{thm:tailapprox}).
(ii) Our head-approximation oracle returns a support set with head value at least a constant fraction of the optimal head value.
Moreover, the EMD-budget of the solution is $O(B \log \frac{k}{w})$ (Theorem~\ref{thm:headapproxbasic}).
Combining these algorithms into our new framework, we obtain a compressive sensing scheme for the CEMD model using $O(k \log(\frac{B}{k} \log(\frac{k}{w})))$ measurements for robust signal recovery.
For a reasonable choice of parameters, e.g., $B = O(k)$, the bound specializes to $m = O(k \log \log(\frac{k}{w}))$, which is very close to the information-theoretic optimum of $m = O(k)$.

\subsection{Prior Work}
Prior to this paper, several efforts have been made to enable compressive sensing recovery for structured sparse signals with approximate projection oracles.
The paper~\cite{Blum11} discusses a Projected Landweber-type method that succeeds even when the projection oracle is approximate.
However, the author assumes that the projection oracle provides an {\em $\epsilon$-additive} tail approximation guarantee.
In other words, for any given $x \in \R^n$, the model-approximation oracle returns a $\widehat{x} \in \Mmodel$ satisfying:
\begin{equation}
\norm{ x - \widehat{x}}_2 = \min_{x' \in \Mmodel} \norm{x - x'}_2 + \eps
\label{eq:addapprox}
\end{equation}
for some parameter $\eps > 0$.
Under such conditions, there exists an algorithm that returns a signal within an
$O(\epsilon)$-neighborhood of the optimal solution.
However, approximation oracles that achieve low additive approximation guarantees satisfying \eqref{eq:addapprox} are rather rare.

On the other hand, the works~\cite{KC,clash} assume the existence of a head-approximation oracle similar to our definition \eqref{eq:head} and develop corresponding signal recovery algorithms. However, these approaches only provide signal recovery guarantees with an additive error term of $O(\|x_{\Omega}\|)$, where $\Omega$ is the set of the $k$ largest coefficients in $x$. Therefore, this result is not directly comparable to our desired recovery guarantee~\eqref{e:lplq}.

Some more recent works have introduced the use of approximate projection oracles, albeit for a different type of signal model.
There, the underlying assumption is that the signals of interest are sparse in a \emph{redundant dictionary}.
The paper~\cite{GE} presents a sparse recovery algorithm for redundant dictionaries that succeeds with multiplicative approximation guarantees.
However, their framework uses only the tail oracle and therefore is subject to the lower bound that we provide in Section~\ref{sec:counterexample}.
In particular, their guarantees make stringent assumptions on the maximum singular values of the sensing matrix $A$. 

The paper~\cite{davenport_redundant} introduces an algorithm called \emph{Signal Space CoSaMP} (SSCoSaMP), which also assumes the existence of multiplicative approximate oracles.
However, the assumptions made on the oracles are restrictive.
Interpreted in the model-based context, the oracles must capture a significant fraction of the optimal support in each iteration, which can be hard to achieve in practice.
The more recent paper~\cite{GN13} proposes a version of SSCoSaMP which succeeds with oracles satisfying both multiplicative head- and tail-approximation guarantees.
Indeed, our AM-CoSaMP algorithm and associated proofs are closely related to this work.
However, AM-CoSaMP requires technically weaker conditions to succeed and our proof techniques are somewhat more concise.
See Section~\ref{sec:amcosamp} for a more detailed discussion on this topic. 

In a parallel line of research, there have been several proposals for compressive sensing methods using {\em sparse} measurement matrices~\cite{BGIK09,cs_sparse_matrices_survey}.
Recent efforts have extended this line of work into the model-based setting.
The paper~\cite{IR13} establishes both lower and upper bounds on the number of measurements required to satisfy the model RIP-1 for certain structured sparsity models.
Assuming that the measurement matrix $A$ satisfies the model RIP-1,
the paper~\cite{BBC14} proposes a modification of \emph{expander iterative hard thresholding} (EIHT)~\cite{RF13}, which achieves stable recovery for arbitrary structured sparsity models.
As with the other algorithms for model-based compressive sensing, EIHT only works with \emph{exact} model projection oracles.
In Section~\ref{sec:modelrip1}, we propose a more general algorithm suitable for model-based recovery using only approximate projection oracles.

We instantiate our algorithmic results in the context of the Constrained Earth Mover's Distance (CEMD) model, developed in~\cite{HIS}. 
The model was originally motivated by the task of reconstructing time
sequences of spatially sparse signals. There has been a substantial
amount of work devoted to such signals, e.g., \cite{modifiedcs,duartedcs}.
We refer the reader to~\cite{HIS} for a more detailed discussion about the model and its applications.
The paper introduced a tail oracle for the problem and empirically
evaluated the performance of the recovery scheme.
Although the use of the oracle was heuristic, the experiments demonstrate a substantial reduction in the number of measurements needed to recover slowly
varying signals.
In this paper, we provide a rigorous analysis of the tail-approximation oracle originally proposed in~\cite{HIS}, as well as a novel head-approximation algorithm.
Combining these two sub-routines yields a model-based compressive sensing scheme for the CEMD model using a nearly optimal number of measurements.

\subsection{Subsequent Work}

Since the appearance of the conference version of this manuscript~\cite{approxSODA}, a number of works have explored some of its implications. The works~\cite{treesISIT,modelcsICALP} develop approximation algorithms for the \emph{tree-sparsity} model~\cite{modelcs}. These algorithms, coupled with our framework, immediately imply sample-optimal recovery schemes for tree-sparse signals that run in \emph{nearly linear-time}.  Additionally, our approximation oracles for the CEMD model can be of independent interest in signal processing applications. For instance,~\cite{faultICASSP} uses the tail-approximation procedure developed in Section~\ref{sec:tail} for detecting \emph{faults} in subsurface seismic images. Investigations into further extensions are currently underway.

\section{Preliminaries}
\label{sec:prelims}
We write $[n]$ to denote the set $\{1, 2, \ldots, n\}$ and $\powerset(A)$ to denote the power set of a set $A$.
For a vector $x \in \R^n$ and a set $\Omega \subseteq [n]$, we write $x_\Omega$ for the restriction of $x$ to $\Omega$, i.e., $(x_\Omega)_i = x_i$ for $i \in \Omega$ and $(x_\Omega)_i = 0$ otherwise.
Similarly, we write $X_\Omega$ for the submatrix of a matrix $X \in \R^{m \times n}$ containing the columns corresponding to $\Omega$, i.e., a matrix in $\R^{m \times \abs{\Omega}}$.
Sometimes, we also restrict a matrix element-wise: for a set $\Omega \subseteq [m] \times [n]$, the matrix $X_\Omega$ is identical to $X$ but the entries not contained in $\Omega$ are set to zero.
The distinction between these two conventions will be clear from context.

A vector $x \in \R^n$ is said to be $k$-sparse if at most $k \leq n$ coordinates are nonzero.
The support of $x$, $\supp(x) \subseteq [n]$, is the set of indices with nonzero entries in $x$.
Hence $x_{\supp(x)} = x$.
Observe that the set of \emph{all} $k$-sparse signals is geometrically equivalent to the union of the $\binom{n}{k}$ canonical $k$-dimensional subspaces of $\R^n$.
For a matrix $X \in \R^{h \times w}$, the support $\supp(X) \subseteq [h] \times [w]$ is also the set of indices corresponding to nonzero entries.
For a matrix support set $\Omega$, we denote the support of a column $c$ in $\Omega$ with $\colsupp(\Omega, c) = \{ r \, | \, (r,c) \in \Omega \}$.

Often, some prior information about the support of a sparse signal $x$ is available.
A flexible way to model such prior information is to consider only the $k$-sparse signals with a permitted configuration of $\supp(x)$.
This restriction motivates the notion of a {\em structured sparsity model}, which is geometrically equivalent to a subset of the $\binom{n}{k}$ canonical $k$-dimensional subspaces of $\R^n$.
\begin{definition}[Structured sparsity model. From Definition 2 in~\cite{modelcs}]
A structured sparsity model $\Mmodel \subseteq \R^n$ is the set of vectors
$\Mmodel = \{ x \in \R^n \, | \, \supp(x) \subseteq S \text{ for some } S \in \Msupports \}$,
where $\Msupports = \{\Omega_1, \ldots, \Omega_l\}$ is the set of allowed structured supports with $\Omega_i \subseteq [n]$.
We call $l = \abs{\Msupports}$ the size of the model $\Mmodel$.
\end{definition}

Note that the $\Omega_i$ in the definition above can have different cardinalities, but the largest cardinality will dictate the sample complexity in our bounds.
Often it is convenient to work with the closure of $\Msupports$ under taking subsets, which we denote with $\Msupportsclosure = \{ \Omega \subseteq [n] \, | \, \Omega \subseteq S \text{ for some } S \in \Msupports \}$.
Then we can write the set of signals in the model as $\Mmodel = \{ x \in \R^n \, | \, \supp(x) \in \Msupportsclosure \}$.

In the analysis of our algorithms, we also use the notion of \emph{model addition}: given two structured sparsity models $\Amodel$ and $\Bmodel$, we define the sum $\Cmodel = \Amodel \supportplus \Bmodel$ as $\Cmodel = \{a + b \, | \, a \in \Amodel \text{ and } b \in \Bmodel \}$ (i.e., the Minkowski sum).
Similarly, we define the corresponding set of allowed supports as $\Csupports = \Asupports \supportplus \Bsupports = \{ \Omega \cup \Gamma \, | \, \Omega \in \Asupports \text{ and } \Gamma \in \Bsupports \}$.
We also use $\Csupports^{\supportplus t}$ as a shorthand for $t$-times addition, i.e., $\Csupports \supportplus \Csupports \supportplus \ldots \supportplus \Csupports$.

The framework of {model-based compressive sensing}~\cite{modelcs} leverages the above notion of a structured sparsity model to design robust sparse recovery schemes. 
Specifically, the framework states that it is possible to recover a structured sparse signal $x \in \Mmodel$ from linear measurements $y = Ax + e$, provided that two conditions are satisfied:
(i) the matrix $A$ satisfies a variant of the restricted isometry property known as the {\em model-RIP}, and
(ii) there exists an oracle that can efficiently {\em project} an arbitrary signal in $\R^n$ onto the model $\Mmodel$. We formalize these conditions as follows.

\begin{definition}[Model-RIP. From Definition 3 in~\cite{modelcs}]
\label{def:modelrip}
The matrix $A\in\R^{m\times n}$ has the $(\delta, \Msupports)$-model-RIP if the following inequalities hold for all $x$ with $\supp(x) \in \Msupports^+$:
\begin{equation}
\label{eq:modelrip}
(1 - \delta) \norm{x}^2_2 \, \leq \, \norm{A x}^2_2 \leq (1 + \delta) \norm{x}^2_2 \, .
\end{equation}
\end{definition}

The following properties are direct consequences of the model-RIP and will prove useful in our proofs in Sections \ref{sec:amiht} and \ref{sec:amcosamp}.
\begin{fact}[adapted from Section 3 in~\cite{cosamp}]
\label{fact:modelrip}
Let $A\in\R^{m \times n}$ be a matrix satisfying the $(\delta, \Msupports)$-model-RIP.
Moreover, let $\Omega$ be a support in the model, i.e., $\Omega \in \Msupports^+$.
Then the following properties hold for all $x \in \R^n$ and $y \in \R^m$:
\begin{align*}
\norm{A^T_\Omega y}_2 & \leq \sqrt{1+\delta}\norm{y}_2 \, ,\\
\norm{A^T_\Omega A_\Omega x}_2 & \leq (1 + \delta) \norm{x}_2 \, ,\\
\norm{\parens{I - A^T_\Omega A_\Omega}x}_2 & \leq \delta \norm{x}_2 \,.
\end{align*}
\end{fact}

\begin{definition}[Model-projection oracle. From Section 3.2 in~\cite{modelcs}]
\label{def:projoracle}
A model-projection oracle is a function $M:\R^n \rightarrow \powerset([n])$ such that the following two properties hold for all $x \in \R^n$.
\begin{description}[font=\normalfont\itshape,nosep]
  \item[Output model sparsity:] $M(x) \in \Msupportsclosure$.
  \item[Optimal model projection:] Let $\Omega' = M(x)$. Then $\norm{x - x_{\Omega'}}_2 = \min_{\Omega \in \Msupports} \norm{x - x_\Omega}_2$.
\end{description}
\end{definition}

Sometimes, we use a model-projection oracle $M$ as a function from $\R^n$ to $\R^n$.
This can be seen as a simple extension of Definition \ref{def:projoracle} where $M(x) = x_\Omega$, $\Omega = M'(x)$, and $M'$ satisfies Definition \ref{def:projoracle}.

Under these conditions, the authors of~\cite{modelcs} show that compressive sampling matching pursuit (CoSaMP~\cite{cosamp}) and iterative hard thresholding
(IHT~\cite{iht}) --- two popular algorithms for sparse recovery --- can be
modified to achieve robust sparse recovery for the model $\Mmodel$.
In particular, the modified version of IHT (called \emph{Model-IHT}~\cite{modelcs}) executes the following iterations until convergence:
\begin{equation}
\label{eq:iht}
x^{i+1} \gets M(x^i + A^T (y - A x^i)) \, ,
\end{equation}
where $x^1=0$ is the initial signal estimate.
From a sampling complexity perspective, the benefit of this approach stems from the model-RIP assumption.
Indeed, the following result indicates that with high probability, a large class of measurement matrices $A$ satisfies the model-RIP with a \emph{nearly optimal} number of rows:
\begin{fact}[\cite{samplingunion,modelcs}]
\label{fact:modelripbound}
Let $\Msupports$ be a structured sparsity model and let $k$ be the size of the largest support in the model, i.e., $k = \max_{\Omega \in \Msupports} \abs{\Omega}$.
Let $A\in\R^{m \times n}$ be a matrix with i.i.d.\ sub-Gaussian entries.
Then there is a constant $c$ such that for $0 < \delta < 1$, any $t > 0$, and
\[
  m \; \geq \; \frac{c}{\delta^2} \,\parens{k \log \frac{1}{\delta} + \log \abs*{\Msupports} + t} \; ,
\]
$A$ has the $(\delta, \Msupports)$-model-RIP with probability at least $1 - e^{-t}$.
\end{fact}
Since $\delta$ and $t$ are typically constants, this bound can often be summarized as
\[
  m \; = \; O(k + \log \abs*{\Msupports}) \; .
\]
If the number of permissible supports (or equivalently, subspaces) $\abs*{\Msupports}$ is asymptotically smaller than $\binom{n}{k}$, then $m$ can be smaller than the $O(k \log \frac{n}{k})$ measurement bound from ``standard'' compressive sensing.
In the ideal case, we have $m = \poly(n) \cdot 2^{O(k)}$, which implies a measurement bound of $m = O(k)$ under the very mild assumption that $k = \Omega(\log n)$.
Since $m = k$ measurements are necessary to reconstruct any $k$-sparse signal, this asymptotic behavior of $m$ is information-theoretically optimal up to constant factors.

While model-based recovery approaches improve upon ``standard'' sparsity-based approaches in terms of sample-complexity, the computational cost of signal recovery crucially depends on the model-projection oracle $M$.
Observe that Model-IHT (Equation \ref{eq:iht}) involves one invocation of 
the model-projection oracle $M$ per iteration, and hence its overall running time scales with that of $M$.
Therefore, model-based recovery approaches are relevant \emph{only} in situations where efficient algorithms for finding the optimal model-projection are available. 

\section{A Negative Result}\label{sec:counterexample}
For many structured sparsity models, computing an optimal model-projection can be a challenging task.
One way to mitigate this computational burden is to use \emph{approximate} model-projection oracles, i.e., oracles that solve the model-projection problem only approximately.
However, in this section we show that such oracles cannot be integrated into Model-IHT (Equation \ref{eq:iht}) in a straightforward manner.

Consider the standard compressive sensing setting, where the ``model'' consists of the set of all $k$-sparse signals.
Of course, finding the optimal model projection in this case is simple: for any signal $x$, the oracle $T_k(\cdot)$ returns the $k$ largest coefficients of $x$ in terms of absolute value.
But for illustrative purposes, let us consider a slightly different
oracle that is approximate in the following sense.
Let $c$ be an arbitrary constant and let $T'_k$ be a projection oracle such that for any $a \in \R^n$ we have:
\begin{equation}
  \label{eq:adversarialguarantee}
  \norm{a - T'_k(a)}_2 \leq c \norm{a - T_k(a)}_2 \, .
\end{equation}
We show that we can construct an ``adversarial'' approximation oracle $T'_k$ that always returns $T'_k(a) = 0$ but still satisfies \eqref{eq:adversarialguarantee} for all signals $a$ encountered during the execution of Model-IHT.
In particular, we use this oracle in Model-IHT and start with the initial signal estimate $x^0 = 0$. 
We will show that such an adversarial oracle still satisfies \eqref{eq:adversarialguarantee} for the first iteration of Model-IHT.
As a result, Model-IHT with this adversarial oracle remains stuck at the zero signal estimate and cannot recover the true signal.

Recall that Model-IHT with projection oracle $T'_k$ iterates
\begin{equation}
\label{eq:bad_miht}
  x^{i+1} \leftarrow T'_k(x^i + A^T (y - A x^i)) \, ,
\end{equation}
which in the first iteration gives
\begin{equation*}
  x^1 \leftarrow T'_k(A^T y) \, .
\end{equation*}
Consider the simplest case where the signal $x$ is $1$-sparse with $x_1 = 1$ and $x_i = 0$ for $i \neq 1$, i.e., $x = e_1$.
Given a measurement matrix $A$ with $(\delta, O(1))$-RIP for small $\delta$, Model-IHT needs to perfectly recover $x$ from $A x$.
It is known that random matrices $A\in\R^{m\times n}$ with $A_{i,j} =
\pm 1/\sqrt{m}$ chosen i.i.d.\ uniformly at random satisfy this RIP for $m =
O(\log n)$ with high probability \cite{jlcs}.\footnote{These are the so-called \emph{Rademacher} matrices.}
We prove that our ``adversarial''  oracle $T'_k(a) = 0$ satisfies the approximation guarantee \eqref{eq:adversarialguarantee} for its input $a = A^T y = A^T A e_1$ with high probability.
Hence, $x^1 = x^0 = 0$ and Model-IHT cannot make
progress.
Intuitively, the tail $a - T_k(a)$ contains so much ``noise'' that
the adversarial approximation oracle $T'_k$ does not need to find a good
sparse support for $a$ and can simply return a signal estimate of $0$.

Consider the components of the vector $a = A^T A e_1$: $a_i$ is the inner product of the first column of $A$ with the $i$-th column of $A$.
Clearly, we have $a_1 = 1$ and $-1 \leq a_i \leq 1$ for $i \neq 1$.
Therefore, $T_k(a) = e_1$ is an optimal projection and $\norm{a - T_k(a)}_2^2 = \norm{a}_2^2 - 1$.
In order to show that the adversarial oracle $T'_k(a)$ satisfies the guarantee \eqref{eq:adversarialguarantee} with constant $c$, we need to prove that:
\begin{equation*}
  \norm{a}_2^2 \leq c^2 (\norm{a}_2^2 - 1) \; .
\end{equation*}
Therefore, it suffices to show that $\norm{a}_2^2 \geq \frac{c^2}{c^2 -1}$. 
Observe that $\norm{a}_2^2 = 1 + \sum_{i=2}^n a_i^2$, where the $a_i$ are independent.
For $i \neq 1$, each $a_i$ is the sum of $m$ independent $\pm\frac{1}{m}$ random variables (with $p = 1/2$) and so $\mathbb{E}[a_i^2] = \frac{1}{m}$.
We can use Hoeffding's inequality to show that $\sum_{i=2}^n a_i^2$ does not deviate from its mean $\frac{n-1}{m}$ by more than $O(\sqrt{n \log n})$ with high probability.
Since $m = O(\log n)$, this shows that for any constant $c > 1$, we will have
\begin{equation*}
\norm{a}_2^2 = 1 + \sum_{i=2}^n a_i^2 \geq \frac{c^2}{c^2 - 1}
\end{equation*}
with high probability for sufficiently large $n$.

Therefore, we have shown that \eqref{eq:bad_miht} does \emph{not} result in a model-based signal recovery algorithm with provable convergence to the correct result $x$.
In the rest of this paper, we develop several alternative approaches that do achieve convergence to the correct result while using approximate projection-oracles.

\section{Approximate Model-IHT}
\label{sec:amiht}

We now introduce our \emph{approximation-tolerant model-based
  compressive sensing} framework. Essentially, we extend the model-based compressive sensing framework to
work with approximate projection oracles, which we formalize in the definitions below.
This extension enables model-based compressive sensing in cases where optimal model projections are beyond our reach, but approximate projections are still efficiently computable. 

The core idea of our framework is to utilize \emph{two} different notions of
approximate projection oracles, defined as follows.

\begin{definition}[Head approximation oracle]
\label{def:headapproxoracle}
  Let $\Msupports, \Msupports_H \subseteq \powerset([n])$, $p \geq 1$, and $c_H \in \R$.
  Then $H : \R^n \rightarrow \powerset([n])$ is a $(c_H, \Msupports, \Msupports_H, p)$-head-approximation oracle if the following two properties hold for all $x \in \R^n$:
  \begin{description}[font=\normalfont\itshape,nosep]
    \item[Output model sparsity:] $H(x) \in \Msupportsclosure_H$.
    \item[Head approximation:] Let $\Omega' = H(x)$.
        Then $\norm{x_{\Omega'}}_p \geq c_H \norm{x_\Omega}_p$ for all $\Omega \in \Msupports$.
  \end{description}
\end{definition}
\begin{definition}[Tail approximation oracle]\label{def:tailapproxoracle}
  Let $\Msupports, \Msupports_T \subseteq \powerset([n])$, $p \geq 1$ and $c_T \in \R$.
  Then $T : \R^n \rightarrow \powerset([n])$ is a $(c_T, \Msupports, \Msupports_T, p)$-tail-approximation oracle if the following two properties hold for all $x \in \R^n$:
  \begin{description}[font=\normalfont\itshape,nosep]
    \item[Output model sparsity:] $T(x) \in \Msupportsclosure_T$.
    \item[Tail approximation:] Let $\Omega' = T(x)$.
        Then $\norm{x - x_{\Omega'}}_p \leq c_T \norm{x - x_\Omega}_p$ for all $\Omega \in \Msupports$.
  \end{description}
\end{definition}

We trivially observe that a head approximation oracle with approximation factor $c_H = 1$ is equivalent to a tail approximation oracle with factor $c_T = 1$, and vice versa.
Further, we observe that for any model $\Mmodel$, if $x \in \Mmodel$ then $\norm{x - x_\Omega}_2 = 0$ for some $\Omega \in \Msupports$. Hence, \emph{any} tail approximation oracle must be exact in the sense that the returned support $\Omega'$ has to satisfy $\norm{x - x_{\Omega'}}_2 = 0$, or equivalently, $\supp(x) \subseteq T(x)$.
On the other hand, we note that $H(x)$ does not need to return an optimal support if the input signal $x$ is in the model $\Mmodel$.

An important feature of the above definitions of approximation oracles is that they permit projections into \emph{larger} models.
In other words, the oracle can potentially return a signal that
belongs to a larger model $\Msupports' \supseteq \Msupports$.
For example, a tail-approximation oracle for the CEMD model with parameters $(k,B)$ is allowed to return a signal with parameters $(2k,2B)$, thereby relaxing both the sparsity constraint and the EMD-budget. We exploit this feature in our algorithms in Section \ref{s:emd}.

Equipped with these notions of approximate projection oracles, we introduce a new algorithm for model-based compressive sensing.
We call our algorithm \emph{Approximate Model-IHT} (AM-IHT); see Algorithm~\ref{alg:approxmodeliht} for a full description.
Notice that every iteration of AM-IHT uses \emph{both} a head-approximation oracle $H$ and a tail-approximation oracle $T$.
This is in contrast to the Model-IHT algorithm discussed above in Section~\ref{sec:counterexample}, which solely made use of a tail approximation oracle $T'$. 

\begin{algorithm}[!t]
\caption{Approximate Model-IHT}
\label{alg:approxmodeliht}
\begin{algorithmic}[1]
\Function{AM-IHT}{$y, A, t$}
\State $x^0 \gets 0$
\For{$i \gets 0, \ldots, t$}
\State $b^i \gets A^T (y - A x^i)$
\State $x^{i+1} \gets T(x^i + H(b^i))$
\EndFor
\State \textbf{return} $x^{t+1}$
\EndFunction
\end{algorithmic}
\end{algorithm}

Our main result of this section (Theorem~\ref{thm:amiht}) states the following:
if the measurement matrix $A$ satisfies the model-RIP for $\Msupports \supportplus \Msupports_T \supportplus \Msupports_H$ and approximate projection oracles $H$ and $T$ are available, then AM-IHT exhibits provably robust recovery.
We make the following assumptions in the analysis of AM-IHT:
(i) $x \in \R^n$ and $x \in \Mmodel$.
(ii) $y = A x + e$ for an arbitrary $e \in \R^m$ (the measurement noise).
(iii) $T$ is a $(c_T, \Msupports, \Msupports_T, 2)$-tail-approximation oracle.
(iv) $H$ is a $(c_H, \Msupports_T \supportplus \Msupports, \Msupports_H, 2)$-head-approximation-oracle.
(v) $A$ has the $(\delta, \Msupports \supportplus \Msupports_T \supportplus \Msupports_H)$-model-RIP.

As in IHT, we use the \emph{residual proxy} $b^i = A^T(y - A x^i)$ as the update in each iteration (see Algorithm \ref{alg:approxmodeliht}). 
The key idea of our proof is the following: when applied to the residual proxy $b^i$,  the head-approximation oracle $H$ returns a support $\Gamma$ that contains ``most'' of the relevant mass contained in $r^i$.
Before we formalize this statement in Lemma \ref{lem:resid}, we first establish the RIP of $A$ on all relevant vectors.
\begin{lemma}\label{lemma:headiht_rip}
Let $r^i = x - x^i$, $\Omega = \supp(r^i)$, and $\Gamma = \supp(H(b^i))$.
For all $x'\in\R^n$ with $\supp(x') \subseteq \Omega \cup \Gamma$ we have
\begin{equation*}
  (1 - \delta)\norm{x'}_2^2 \leq \norm{A x'}_2^2 \leq (1 + \delta)\norm{x'}_2^2 \; .
\end{equation*}
\end{lemma}
\begin{proof}
By the definition of $T$, we have $\supp(x^i) \in \Msupports_T$.
Since $\supp(x) \in \Msupports$, we have $\supp(x - x^i) \in \Msupports_T \supportplus \Msupports$ and hence $\Omega \in \Msupports_T \supportplus \Msupports$.
Moreover, $\supp(H(b^i)) \in \Msupports_H$ by the definition of $H$.
Therefore $\Omega \cup \Gamma \in \Msupports \supportplus \Msupports_T \supportplus \Msupports_H$, which allows us to use the model-RIP of $A$ on $x'$ with $\supp(x') \subseteq \Omega \cup \Gamma$.
\end{proof}

We now establish our main lemma, which will also prove useful in Section \ref{sec:amcosamp}.
A similar result (with a different derivation approach and different constants) appears in Section 4 of the conference version of this manuscript~\cite{approxSODA}.

\begin{lemma}
\label{lem:resid}
Let $r^i = x - x^i$ and $\Gamma = \supp(H(b^i))$.
Then, 
\begin{equation}
\norm{r^i_{\Gamma^c}}_2  \, \leq \,  \sqrt{1 - \alpha_0^2}\norm{r^i}_2 +
\left[\frac{\beta_0}{\alpha_0} + \frac{\alpha_0 \beta_0}{\sqrt{1-\alpha_0^2}} \right] \norm{e}_2 \, .
\label{eq:rlambdac}
\end{equation}
where
$$
\alpha_0 = c_H (1 - \delta) - \delta \qquad\text{and} \qquad\beta_0 = (1+c_H)\sqrt{1+\delta} \, . 
$$
We assume that $c_H$ and $\delta$ are such that $\alpha_0 > 0$.
\end{lemma}
\begin{proof}
We provide lower and upper bounds on $\norm*{H(b^i)}_2 = \norm{b^i_\Gamma}_2$, where $b^i = A^T(y - Ax_i) = A^T A r^i + A^T e$.
Let $\Omega = \supp(r^i)$.
From the head-approximation property, we can bound $ \norm{b^i_\Gamma}_2$ as:
\begin{align*}
\norm{b^i_\Gamma}_2 &= \norm{A_\Gamma^T A r^i + A_\Gamma^T e}_2 \\
    & \geq  c_H \norm{A_\Omega^T A r^i + A_\Omega^T e}_2 \\
    & \geq c_H \norm{A_\Omega^T A_\Omega r^i }_2 - c_H \norm{A_\Omega^T e}_2 \\
    & \geq  c_H (1 - \delta) \norm{r^i}_2  - c_H \sqrt{1 + \delta} \norm{e}_2 \, ,
\end{align*}
where the inequalities follow from Fact~\ref{fact:modelrip} and the triangle inequality.
This provides the lower bound on $\norm{b^i_\Gamma}_2$.

Now, consider $r_\Gamma$. By repeated use of the triangle inequality, we get
\begin{align*}
\norm{b^i_\Gamma}_2 &= \norm{A_\Gamma^T A r^i + A_\Gamma^T e}_2 \\
    &= \norm{A_\Gamma^T A r^i - r^i_\Gamma + r^i_\Gamma + A_\Gamma^T e}_2 \\
    &\leq \norm{A_\Gamma^T A r^i - r^i_\Gamma}_2 + \norm{r^i_\Gamma}_2 + \norm{A_\Gamma^T e}_2 \\
    &\leq \norm{A_{\Gamma \cup \Omega}^T A r^i - r^i_{\Gamma \cup \Omega}}_2 + \norm{r^i_\Gamma}_2 + \sqrt{1+\delta}\norm{e}_2 \\
    &\leq \delta \norm{r^i}_2 + \norm{r^i_\Gamma}_2 + \sqrt{1+\delta}\norm{e}_2 \, ,
\end{align*}
where the last inequality again follows from Fact~\ref{fact:modelrip}.
This provides the upper bound on  $\norm{b^i_\Gamma}_2$.

Combining the two bounds and grouping terms, we obtain the following inequality.
In order to simplify notation, we write $\alpha_0 = c_H (1 - \delta) - \delta$
and $\beta_0 = (1+c_H)\sqrt{1+\delta}$. 
\begin{equation}
\label{eq:lower_head}
\norm{r^i_\Gamma}_2 \geq \alpha_0 \norm{r^i}_2 - \beta_0 \norm{e}_2 \, .
\end{equation}
Next, we examine the right hand side of \eqref{eq:lower_head} more carefully.
Let us assume that the RIP constant $\delta$ is set to be small enough such that it satisfies $c_H > \delta/(1-\delta)$.
There are two mutually exclusive cases:

\noindent{\emph{Case 1:}} The value of $\norm*{r^i}_2$ satisfies $\alpha_0 \norm*{r^i}_2 \leq \beta_0 \norm{e}_2$.
Then, consider the vector $r^i_{\Gamma^c}$, i.e., the vector $r^i$ restricted to the set of coordinates in the complement of $\Gamma$.
Clearly, its norm is smaller than $\norm{r^i}_2$.
Therefore, we have
\begin{equation}
\norm{r^i_{\Gamma^c}}_2 \leq \frac{\beta_0}{\alpha_0} \norm{e}_2 \, . \label{eq:iht_case1}
\end{equation}
\noindent{\emph{Case 2:}} The value of $\norm*{r^i}_2$ satisfies $\alpha_0 \norm*{r^i}_2 \geq \beta_0 \norm{e}_2$.
Rewriting \eqref{eq:lower_head}, we get
\[
\norm{r^i_\Gamma}_2 \geq \norm{r^i}_2 \left(\alpha_0 - \frac{\beta_0 \norm{e}_2}{\norm{r_i}_2} \right) \, . 
\]
Moreover, we also have $\norm{r^i}_2^2 = \norm{r^i_\Gamma}_2^2 +
\norm{r^i_{\Gamma^c}}_2^2$. 
Therefore, we obtain
\begin{equation}
\norm{r^i_{\Gamma^c}}_2  \leq \norm{r^i}_2  \sqrt{1 - \left(\alpha_0 - \beta_0 \frac{\norm{e}_2}{\norm{r^i}_2} \right)^2}\, .
\label{eq:iht_case2_prelim}
\end{equation}

We can simplify the right hand side using the following  geometric
argument, adapted from~\cite{lee2010admira}.
Denote $\omega_0= \alpha_0 - \beta_0 \norm{e}_2/\norm{r^i}_2 $.
Then, $0 \leq \omega_0 < 1$ because $\alpha_0 \norm*{r^i}_2 \geq \beta_0 \norm{e}_2$, $\alpha_0 < 1$, and $\beta_0 \geq 1$.
The function $g(\omega_0) = \sqrt{1-\omega_0^2}$ traces an arc of the unit circle as a function of $\omega_0$ and therefore is upper-bounded by the $y$-coordinate of \emph{any} tangent line to the circle evaluated at $\omega_0$.
For a free parameter $0 < \omega < 1$ (the tangent point of the tangent line), a straightforward calculation yields that
\[
\sqrt{1-\omega_0^2} \leq \frac{1}{\sqrt{1-\omega^2}} - \frac{\omega}{\sqrt{1-\omega^2}} \omega_0 \; .
\]
Therefore, substituting into the bound
for $\norm{r^i_{\Gamma^c}}_2$, we get:
\begin{align*}
\norm{r^i_{\Gamma^c}}_2  &\leq \norm{r^i}_2 \left(
  \frac{1}{\sqrt{1-\omega^2}}  - \frac{\omega}{\sqrt{1-\omega^2}} \left( \alpha_0 -
    \beta_0 \frac{\norm{e}_2}{\norm{r^i}_2} \right) \right) \\
&=\frac{1-\omega\alpha_0}{\sqrt{1-\omega^2}} \norm{r^i}_2 + \frac{\omega
  \beta_0}{\sqrt{1-\omega^2}} \norm{e}_2 \; .
\end{align*}
The coefficient preceding $\norm{r^i}_2$ determines the overall convergence rate, and the minimum value of the coefficient is attained by setting
$\omega = \alpha_0$. Substituting, we obtain
\begin{equation}
\norm{r^i_{\Gamma^c}}_2 \leq \sqrt{1-\alpha_0^2} \norm{r^i}_2 +
\frac{\alpha_0 \beta_0}{\sqrt{1-\alpha_0^2}} \norm{e}_2 \, .
\label{eq:iht_case2}
\end{equation}
Combining the mutually exclusive cases \eqref{eq:iht_case1} and
\eqref{eq:iht_case2}, we obtain
\begin{equation*}
\norm{r^i_{\Gamma^c}}_2  \leq  \sqrt{1 - \alpha_0^2}\norm{r^i}_2 +
\left[\frac{\beta_0}{\alpha_0} + \frac{\alpha_0 \beta_0}{\sqrt{1-\alpha_0^2}} \right] \norm{e}_2 \, ,
\end{equation*}
which proves the lemma.
\end{proof}

\begin{theorem}[Geometric convergence of AM-IHT]
\label{thm:amiht}
Let $r^i = x-x^i$, where $x^i$ is the signal estimate computed by AM-IHT in iteration $i$.
Then, 
\begin{equation*}
\norm{r^{i+1}}_2  \leq \alpha \norm{r^i}_2 + \beta \norm{e}_2 \, ,
\end{equation*}
where 
\begin{align*}
\alpha &= (1+c_T)\left[\delta + \sqrt{1 -\alpha_0^2}\right]\, , & \beta &= (1+c_T)\left[\frac{\beta_0}{\alpha_0} + \frac{\alpha_0 \beta_0}{\sqrt{1-\alpha_0^2}} +  \sqrt{1+\delta}\right] \, , \\
\alpha_0 &= c_H (1 - \delta) - \delta\, , & \beta_0 &= (1+c_H)\sqrt{1+\delta} \, .
\end{align*}
We assume that $c_H$ and $\delta$ are such that $\alpha_0 > 0$.
\end{theorem}
\begin{proof}
Let $a = x^i + H(b^i)$.  
From the triangle inequality, we have:
\begin{align*}
\norm{x - x^{i+1}}_2  &= \norm{x - T(a)}_2 \\
    &\leq  \norm{x - a}_2 + \norm{a - T(a)}_2 \\
    &\leq (1 + c_T) \norm{x - a}_2 \\
    &=  (1 + c_T) \norm{x - x^i - H(b^i)}_2  \\
&= (1 + c_T) \norm{r^i - H(A^T A r^i + A^T e)}_2 \numberthis \label{eq:tail_iht} \, .
\end{align*}

We can further bound $\norm{r^i - H(A^T A r^i + A^T e)}_2$ in terms of $\norm{r^i}_2$.
Let $\Omega = \supp(r^i)$ and $\Gamma = \supp(H(A^T A r^i + A^T e))$.
We have the inequalities
\begin{align*}
\norm{r^i - H(A^T A r^i + A^T e)}_2  = &~\norm{r^i_\Gamma +
  r^i_{\Gamma^c} - A_\Gamma^T A r^i + A_\Gamma^T e}_2 \\
 \leq &~\norm{A_\Gamma^T A r^i - r^i_\Gamma}_2 + \norm{r^i_{\Gamma^c}}_2 + \norm{A_\Gamma^T e}_2 \\
 \leq &~\norm{A_{\Gamma \cup \Omega}^T A r^i - r^i_{\Gamma \cup \Omega}}_2 + \norm{r^i_{\Gamma^c}}_2 + \norm{A_\Gamma^T e}_2 \\
 \leq &~\delta \norm{r^i}_2 + \sqrt{1 - \alpha_0^2}\norm{r^i}_2 + \left[\frac{\beta_0}{\alpha_0} + \frac{\alpha_0 \beta_0}{\sqrt{1-\alpha_0^2}} + \sqrt{1+\delta}\right] \norm{e}_2 \, ,
\end{align*}
where the last inequality follows from the RIP and
\eqref{eq:rlambdac}.
Putting this together with \eqref{eq:tail_iht} and grouping terms, we get 
\begin{equation}
\norm{x - x^{i+1}}_2 \leq \alpha \norm{x - x^i}_2 + \beta \norm{e}_2 \, ,
\label{eq:newbound}
\end{equation}
thus proving the Theorem. 
\end{proof}

In the noiseless case, we can ignore the second term and only focus on
the leading recurrence factor: 
$$\alpha = (1 + c_T) \left( \delta + \sqrt{1 - (c_H (1 -
    \delta) - \delta)^2} \right) .$$
For convergence, we need $\alpha$ to be strictly smaller than 1.
Note that we can make $\delta$ as small as we desire since this assumption only affects the measurement bound by a constant factor.
Therefore, the following condition must hold for guaranteed convergence:
\begin{equation}
(1 + c_T) \sqrt{1 - c_H^2} < 1 \, , \qquad \textnormal{or equivalently, } \qquad c_H^2 > 1 - \frac{1}{(1 + c_T)^2} \; .
\label{eq:headtail_iht}
\end{equation}
Under this condition, AM-IHT exhibits geometric convergence comparable to the existing model-based compressive sensing results of~\cite{modelcs}.
AM-IHT achieves this \emph{despite using only approximate projection
oracles}.
In Section \ref{sec:boosting}, we relax condition \eqref{eq:headtail_iht} so that geometric convergence is possible for \emph{any} constants $c_T$ and $c_H$.

The geometric convergence of AM-IHT implies that the algorithm quickly
recovers a good signal estimate. Formally, we obtain:
\begin{corollary}
\label{cor:amiht}
Let $T$ and $H$ be approximate projection oracles with $c_T$ and $c_H$ such that $0 < \alpha < 1$.
Then after $t = \ceil{\frac{\log \frac{\norm{x}_2}{\norm{e}_2}}{\log \frac{1}{\alpha}}}$ iterations, AM-IHT returns a signal estimate $\widehat{x}$ satisfying
\begin{equation*}
  \norm{x - \widehat{x}}_2 \leq \left(1 + \frac{\beta}{1 - \alpha} \right) \norm{e}_2 \; .
\end{equation*}
\end{corollary}
\begin{proof}
As before, let $r^i = x - x^i$.
Using $\norm*{r^0}_2 = \norm{x}_2$, Theorem \ref{thm:amiht}, and a simple inductive argument shows that
\[
  \norm*{r^{i+1}}_2 \leq \alpha^i \norm{x}_2 + \beta \norm{e}_2 \sum_{j=0}^i \alpha^j \, .
\]
For $i = \ceil{\frac{\log \frac{\norm{x}_2}{\norm{e}_2}}{\log \frac{1}{\alpha}}}$, we get $\alpha^i \norm{x}_2 \leq \norm{e}_2$.
Moreover, we can bound the geometric series $\sum_{j = 0}^t \alpha^j$ by $\frac{1}{1-\alpha}$.
Combining these bounds gives the guarantee stated in the theorem.
\end{proof}

\section{Approximate Model-CoSaMP}
\label{sec:amcosamp}
In this Section, we propose a second algorithm for model-based compressive sensing with approximate projection oracles.
Our algorithm is a generalization of model-based CoSaMP, which was initially developed in~\cite{modelcs}.
We call our variant \emph{Approximate Model-CoSaMP} (or AM-CoSaMP); see
Algorithm~\ref{alg:approxmodelcosamp} for a complete description. 

Algorithm~\ref{alg:approxmodelcosamp} closely resembles the {\em Signal-Space
CoSaMP} (or SSCoSaMP) algorithm proposed and
analyzed in~\cite{davenport_redundant,GN13}. Like our approach, SSCoSaMP also makes
assumptions about the existence of head- and tail-approximation
oracles. However, there are some important technical differences in our
development. SSCoSaMP was introduced in the context of
recovering signals that are sparse in overcomplete and incoherent
dictionaries.
In contrast, we focus on recovering signals from structured sparsity models.

Moreover, the authors of~\cite{davenport_redundant,GN13} assume that a \emph{single} oracle simultaneously achieves the conditions specified in Definitions \ref{def:headapproxoracle} and \ref{def:tailapproxoracle}.
In contrast, our approach assumes the existence of two separate head- and tail-approximation oracles and consequently is somewhat more general.
Finally, our analysis is simpler and more concise than that provided in~\cite{davenport_redundant,GN13} and follows directly from the results in Section~\ref{sec:amiht}. 

\begin{algorithm}[!t]
\caption{Approximate Model-CoSaMP}
\label{alg:approxmodelcosamp}
\begin{algorithmic}[1]
\Function{AM-CoSaMP}{$y,A,t$}
\State $x^0 \gets 0$
\For{$i \gets 0, \ldots, t$}
\State $b^i \gets A^T (y - A x^i)$
\State $\Gamma \gets \mathrm{supp}(H(b^i))$
\State $S \gets \Gamma \cup \mathrm{supp}(x^i)$ \label{line:cosampsupportupdate}
\State $z \vert_S \gets A_S\pinv y, \quad z \vert_{S^C} \gets 0$
\State $x^{i+1} \gets T(z)$
\EndFor
\State \textbf{return} $x^{t+1}$
\EndFunction
\end{algorithmic}
\end{algorithm}

We prove that AM-CoSaMP (Alg.\ \ref{alg:approxmodelcosamp}) exhibits robust signal recovery.
We make the same assumptions as in Section \ref{sec:amiht}:
(i) $x \in \R^n$ and $x \in \Mmodel$.
(ii) $y = A x + e$ for an arbitrary $e \in \R^m$ (the measurement noise).
(iii) $T$ is a $(c_T, \Msupports, \Msupports_T, 2)$-tail-approximation oracle.
(iv) $H$ is a $(c_H, \Msupports_T \supportplus \Msupports, \Msupports_H, 2)$-head-approximation-oracle.
(v) $A$ has the $(\delta, \Msupports \supportplus \Msupports_T \supportplus \Msupports_H)$-model-RIP.
Our main result in this section is the following:

\begin{theorem}[Geometric convergence of AM-CoSaMP]
\label{thm:amcosamp}
Let $r^i = x-x^i$, where $x^i$ is the signal estimate computed by AM-CoSaMP in iteration $i$.
Then, 
\begin{equation*}
\norm{r^{i+1}}_2  \leq \alpha \norm{r^i}_2 + \beta \norm{e}_2 \, ,
\end{equation*}
where 
\begin{align*}
\alpha &= (1+c_T) \sqrt{ \frac{1+\delta}{1-\delta}} \sqrt{1
  -\alpha_0^2} \, , \\
\beta &= (1+c_T) \left[\sqrt{ \frac{1+\delta}{1-\delta}} \left(\frac{\beta_0}{\alpha_0} + \frac{\alpha_0 \beta_0}{\sqrt{1-\alpha_0^2}}\right) + \frac{2}{\sqrt{1-\delta}} \right] \, , \\
\alpha_0 &= c_H (1 - \delta) - \delta \, , \\
\beta_0 &= (1+c_H)\sqrt{1+\delta} \, .
\end{align*}

\end{theorem}
\begin{proof}
We can bound the error $\norm*{r^{i+1}}_2$ as follows:
\begin{align*}
\norm{r^{i+1}}_2 &= \norm{x - x^{i+1}}_2 \\
    &\leq  \norm{x^{i+1} - z}_2 + \norm{x - z}_2 \\
    &\leq  c_T \norm{x - z}_2 + \norm{x - z}_2 \\
    &= (1 + c_T) \norm{x - z}_2 \\
    &\leq  (1 + c_T) \frac{\norm{A (x - z )}_2}{\sqrt{1 - \delta}} \\
    &=  (1 + c_T) \frac{\norm{ Ax - A  z }_2}{\sqrt{1 - \delta}} \, .
\end{align*}
Most of these inequalities follow the same steps as the proof provided in~\cite{cosamp}.
The second relation above follows from the triangle inequality, the third relation follows from the tail approximation property and the fifth relation follows from the $(\delta, \Msupports \supportplus \Msupports_T \supportplus \Msupports_H)$-model-RIP of $A$.

We also have $Ax = y - e$ and $Az = A_S z_S$.
Substituting, we get:
\begin{align*}
\norm{r^{i+1}}_2 & \leq   (1 + c_T) \left( \frac{\norm{ y - A_S  z_S }_2}{\sqrt{1 - \delta}} + \frac{\norm{e}_2}{\sqrt{1 -\delta}} \right) \\
& \leq   (1 + c_T) \left( \frac{\norm{ y - A_S  x_S }_2}{\sqrt{1 - \delta}} + \frac{\norm{e}_2}{\sqrt{1 -\delta}} \right) \, . \numberthis \label{eq:cosamperror}
\end{align*}
The first inequality follows from the triangle inequality and the second from the fact that $z_S$ is the least squares estimate $A_S\pinv y $ (in particular, it is at least as good as $x_S$).

Now, observe that $y = Ax +e = A_S x_S + A_{S^c} x_{S^c} + e$.  Therefore, we can further simplify inequality \eqref{eq:cosamperror} as
\begin{align*}
\norm{r^{i+1}}_2 & \leq (1 + c_T) \frac{\norm{ A_{S^c}  x_{S^c}}_2}{\sqrt{1 - \delta}} + (1+c_T) \frac{2\norm{e}_2}{\sqrt{1-\delta}} \\
    &\leq  (1 + c_T) \frac{\sqrt{1 + \delta}}{\sqrt{1 - \delta}} \norm{  x_{S^c} }_2 + (1+c_T) \frac{2\norm{e}_2}{\sqrt{1 -\delta}} \\
    &= (1 + c_T) \sqrt{\frac{{1 + \delta}}{{1 - \delta}}} \norm{ (x - x^i)_{S^c} }_2 + (1+c_T) \frac{2\norm{e}_2}{\sqrt{1 -\delta}} \\
    &\leq (1 + c_T) \sqrt{\frac{{1 + \delta}}{{1 - \delta}}} \norm{  r^i_{\Gamma^c} }_2 + (1+c_T) \frac{2\norm{e}_2}{\sqrt{1 -\delta}} \numberthis \label{eq:ineq1} \, .
\end{align*}
The first relation once again follows from the triangle inequality.
The second relation follows from the fact that $\supp(x_{S^c}) \in \Msupports^+$ (since $\supp(x) \in \Msupports^+$), and therefore, $A_{S^c} x_{S^c}$ can be upper-bounded using the model-RIP.
The third follows from the fact that $x_i$ supported on $S^c$ is zero because $S$ fully subsumes the support of $x^i$.
The final relation follows from the fact that $S^c \subseteq \Gamma^c$ (see line \ref{line:cosampsupportupdate} in the algorithm).

Note that the support $\Gamma$ is defined as in Lemma \ref{lemma:headiht_rip}.
Therefore, we can use \eqref{eq:rlambdac} and bound $\norm{r^i_{\Gamma^c}}_2$ in terms of $\norm{r^i}_2$, $c_H$, and $\delta$.
Substituting into \eqref{eq:ineq1} and rearranging terms, we obtain the stated theorem.
\end{proof}

As in the analysis of AM-IHT, suppose that $e = 0$ and $\delta$ is very small.
Then, we achieve geometric convergence, i.e., $\alpha < 1$, if the approximation factors $c_T$ and $c_H$ satisfy
\begin{equation}
(1 + c_T) \sqrt{1 - c_H^2} < 1 \, , \qquad \textnormal{or equivalently, } \qquad c_H^2 > 1 - \frac{1}{(1 + c_T)^2}\, .
\label{eq:headtail_cosamp}
\end{equation}
Therefore, the conditions for convergence of AM-IHT and AM-CoSaMP are identical in this regime.
As for AM-IHT, we relax this condition for AM-CoSaMP in Section \ref{sec:boosting} and show that geometric convergence is possible for \emph{any} constants $c_T$ and $c_H$.

\section{Approximate Model-IHT with RIP-1 matrices}
\label{sec:modelrip1}

AM-IHT and AM-CoSaMP (Algorithms~\ref{alg:approxmodeliht} and~\ref{alg:approxmodelcosamp}) rely on measurement matrices satisfying the model-RIP (Definition \ref{def:modelrip}).
It is known that $m \times n$ matrices whose elements are drawn i.i.d.\ from a sub-Gaussian distribution satisfy this property with high probability while requiring only a small number of rows $m$ \cite{modelcs,jlcs}.
However, such matrices are \emph{dense} and consequently incur significant costs of $\Theta(m \cdot n)$ for both storage and matrix-vector multiplications. 

One way to circumvent this issue is to consider \emph{sparse} measurement matrices~\cite{cs_sparse_matrices_survey}.
Sparse matrices can be stored very efficiently and enable fast matrix-vector multiplication (with both costs scaling proportionally to the number of nonzeros).
However, the usual RIP does not apply for such matrices.
Instead, such matrices are known to satisfy the RIP in the $\ell_1$-norm (or {\em RIP-1}).
Interestingly, it can be shown that this property is sufficient to enable robust
sparse recovery for arbitrary signals~\cite{BGIK09}.
Moreover, several existing algorithms for sparse recovery can be modified to work with sparse measurement matrices; see~\cite{BGIK09,RF13}.

In the model-based compressive sensing context, one can analogously define the RIP-1 over structured
sparsity models as follows:
\begin{definition}[Model RIP-1]
A matrix $A\in\R^{m\times n}$ has the $(\delta, \Msupports)$-model RIP-1 if the following holds for all $x$ with $\supp(x) \in \Msupports^+$:
\begin{equation}
(1 - \delta) \norm{x}_1 \, \leq \, \norm{A x}_1 \leq (1 + \delta) \norm{x}_1 \, .
\end{equation}
\end{definition}
The paper~\cite{IR13} establishes both lower and upper bounds on
the number of measurements required to satisfy the model RIP-1 for
certain structured sparsity models.
Similar to Fact \ref{fact:modelripbound}, the paper also provides a general sampling bound based on the cardinality of the model:
\begin{fact}[Theorem 9 in \cite{IR13}]
\label{fact:modelripbound1}
Let $\Mmodel$ be a structured sparsity model and let $k$ be the size of the largest support in the model, i.e., $k = \max_{\Omega \in \Msupports} \abs{\Omega}$.
Then there is a $m \times n$ matrix satisfying the $(\delta, \Msupports)$-model RIP-1 with
\[
  m = O\parens{\frac{k}{\delta^2} \cdot \frac{\log (n /l)}{\log (k / l)}} \, ,
\]
where
\[
  l = \frac{\log \abs{\Msupports}}{\log (n/k)} \, .
\]
\end{fact}

Subsequently, the paper~\cite{BBC14} proposes a modification of  {\em expander
  iterative hard thresholding} (EIHT)~\cite{RF13} that achieves stable
recovery for arbitrary structured sparsity
models. As before, this modified algorithm only works when provided
access to {\em exact} model-projection oracles. Below, we propose a
more general algorithm suitable for model-based recovery using only
approximate projection oracles.

Before proceeding further, it is worthwhile to understand a particular class of matrices that satisfy the RIP-1.
It is known that adjacency matrices of certain carefully chosen random bipartite graphs, known as {\em bipartite expanders}, satisfy the model RIP-1~\cite{BGIK09,IR13}.
Indeed, suppose that such a matrix $A$ represents the bipartite graph $G = ([n],[m],E)$, where $E$ is the set of edges. 
For any $S \subseteq [n]$, define
$\Gamma(S)$ to be the set of nodes in $[m]$ connected to $S$ by an edge
in $E$. Therefore, we can define the {\em median operator} $\med{u} :
\reals^m \rightarrow \reals^n$ for any $u \in \R^m$ component-wise as follows:
\[
[\med{u}]_i = \textrm{median} [u_j : j \in \Gamma(\{i\})] \; .
\]
This operator is crucial in our algorithm and proofs below.

\begin{algorithm}[!t]
\caption{AM-IHT with RIP-1}
\label{alg:approxmodeliht_rip1}
\begin{algorithmic}[1]
\Function{AM-IHT-RIP-1}{$y,A,t$}
\State $x^0 \gets 0$
\For{$i \gets 0, \ldots, t$}
\State $x^{i+1} \gets T(x^i + H(\med{y - A x^i}))$
\EndFor
\State \textbf{return} $x^{t+1}$
\EndFunction
\end{algorithmic}
\end{algorithm}

We now propose a variant of AM-IHT (Algorithm~\ref{alg:approxmodeliht}) that is suitable when the measurement matrix $A$ satisfies the RIP-1.
The description of this new version is provided as Algorithm \ref{alg:approxmodeliht_rip1}.
Compared to AM-IHT, the important modification in the RIP-1 algorithm is the use of the median operator $\med{\cdot}$ instead of the transpose of the measurement matrix $A$. 

We analytically characterize the convergence behavior of Algorithm \ref{alg:approxmodeliht_rip1}. 
First, we present the following Lemma, which is proved in \cite{BBC14} based on \cite{RF13}.
\begin{lemma}[Lemma 7.2 in \cite{BBC14}]
\label{lem:rip1bound}
Suppose that $A$ satisfies the $(\delta,\Msupports)$-model-RIP-1. 
Then, for any vectors $x \in \reals^n$, $e \in \reals^m$, and any support $S \in \Msupports^+$, 
\[
\norm{[x - \med{Ax_S + e}]_S}_1 \leq \rho_0  \norm{x_S}_1 + \tau_0 \norm{e}_1 \, .
\]
Here, $\rho_0 = 4\delta/(1-4\delta)$ and $\tau_0$ is a positive scalar that depends
on $\delta$.  
\end{lemma}
Armed with this Lemma, we now prove the main result of this section.
We make similar assumptions as in Section \ref{sec:amiht}, this time using the model-RIP-1 and approximate projection oracles for the $\ell_1$-norm:
(i) $x \in \R^n$ and $x \in \Mmodel$.
(ii) $y = A x + e$ for an arbitrary $e \in \R^m$ (the measurement noise).
(iii) $T$ is a $(c_T, \Msupports, \Msupports_T, 1)$-tail-approximation oracle.
(iv) $H$ is a $(c_H, \Msupports_T \supportplus \Msupports, \Msupports_H, 1)$-head-approximation-oracle.
(v) $A$ has the $(\delta, \Msupports \supportplus \Msupports_T \supportplus \Msupports_H)$-model-RIP-1.
Then, we obtain:

\begin{theorem}[Geometric convergence of AM-IHT with RIP-1]
\label{thm:approxmodeliht_rip1}
Let $r^i = x-x^i$, where $x^i$ is the signal estimate computed by AM-IHT-RIP-1 in iteration $i$.
Let $\rho_0, \tau_0$ be as defined in Lemma~\ref{lem:rip1bound}.
Then, AM-IHT-RIP-1  exhibits the following convergence property:
\[
\norm{r^{i+1}}_1 \leq \rho \norm{r^i}_1 + \tau \norm{e}_1 \, ,
\]
where
\begin{align*}
\rho &= (1+c_T) (2\rho_0 + 1 - c_H (1 - \rho_0)) \, , \\
\tau &= (1+c_T)(2+c_H)\tau_0 \, .
\end{align*}
\end{theorem}
\begin{proof}
Let $a_i = x_i + H(\med{y - A x_i})$. The triangle inequality gives:
\begin{align*}
\norm{r^{i+1}}_1 &= \norm{x - x^{i+1}}_1 \\
              &\leq \norm{x - a^i}_1 + \norm{x^{i+1} - a^i}_1 \\
              &\leq (1+c_T) \norm{x - a^i}_1 \\
              & \leq (1+c_T) \norm{x - x^i - H(\med{y - A x^i})}_1 \\
              &  = (1+c_T) \norm{r^i - H(\med{Ar^i + e})}_1 \, .
\end{align*}
Let $v = \med{ A r^i + e}$, $\Omega = \textrm{supp}(r^i)$, and $\Gamma$ be the support returned by the head oracle $H$.
We have:
\begin{equation}
\norm{H(v)}_1 = \norm{v_\Gamma}_1 \geq c_H \norm{v_\Omega}_1 \, ,
\label{eq:Hv1}
\end{equation}
due to the head-approximation property of $H$. 

On the other hand, we also have 
\begin{align*}
\norm{v_\Omega - r^i}_1 &= \norm{(\med{Ar^i + e} - r^i)_\Omega}_1 \\
&\leq  \norm{(\med{Ar^i + e} - r^i)_{\Omega \cup \Gamma}}_1  \\
&\leq  \rho_0 \norm{r^i}_1 + \tau_0 \norm{e}_1 \, .
\end{align*}
where the last inequality follows from Lemma~\ref{lem:rip1bound} (note that we use the lemma for the model $\Msupports \supportplus \Msupports_T \supportplus \Msupports_H$).
Further, by applying the triangle inequality again and combining with~\eqref{eq:Hv1}, we get
\begin{equation}
\label{eq:head_rip1}
\norm{H(v)}_1 \geq c_H (1 - \rho_0) \norm{r^i}_1 - c_H \tau_0 \norm{e}_1 \, .
\end{equation}
We also have the following series of inequalities:
\begin{align*}
\norm{H(v)}_1 &= \norm{H(v) - r^i_{\Gamma} + r^i_\Gamma}_1 \\
              &\leq \norm{v_\Gamma - r^i_\Gamma}_1 + \norm{r^i_\Gamma}_1 \\
              &\leq \norm{v_{\Gamma \cup \Omega}-r^i_{\Gamma \cup
                  \Omega}}_1 + \norm{r^i_\Gamma}_1 \\
& = \norm{(\med{Ar^i + e} - r^i)_{\Omega \cup \Gamma}}_1 +
\norm{r^i_\Gamma}_1 \\
&\leq \rho_0 \norm{r^i}_1 + \tau_0 \norm{e}_1 + \norm{r^i_\Gamma}_1 \, .
 \end{align*}
Here, we have once again invoked
Lemma~\ref{lem:rip1bound}. Moreover, $\norm{r^i_\Gamma}_1 = \norm{r^i}_1 -
\norm{r^i_{\Gamma^c}}_1$. Combining with \eqref{eq:head_rip1} and rearranging terms, we get:
\begin{equation}
\label{eq:comp}
\norm{r^i_{\Gamma^c}}_1 \leq (\rho_0 + 1 - c_H (1 - \rho_0))\norm{r^i}_1 + (1+c_H) \tau_0 \norm{e}_1 \, .
\end{equation}
Recall that 
\begin{align*}
\norm{r^{i+1}}_1 &\leq (1+c_T)\norm{r^i - H(v)}_1 \\
&= (1+c_T) \left( \norm{r^i_\Gamma - v_\Gamma}_1 + \norm{r^i_{\Gamma^c}}_1 \right) ,
\end{align*}
since $v_\Gamma = H(v) = H(\med{Ar^i + e})$. Invoking Lemma~\ref{lem:rip1bound}
one last time and combining with \eqref{eq:comp}, we obtain
\begin{align*}
\norm{r^{i+1}}_1 &\leq (1+c_T) \left[ \rho_0 \norm{r^i}_1 + \tau_0 \norm{e}_1 + (\rho_0
+ 1 - c_H(1 - \rho_0))\norm{r^i}_1 + (1+c_H) \tau_0 \norm{e}_1 \right] \\
&\leq (1+c_T) (2\rho_0 + 1 - c_H (1 - \rho_0)) \norm{r^i}_1 +
(1+c_T)(2+c_H)\tau_0 \norm{e}_1 \, , 
\end{align*}
 as claimed.
\end{proof}

Once again, if $e = 0$ and $\rho_0$ is made sufficiently small, AM-IHT with RIP-1 achieves geometric convergence to the
true signal $x$ provided that $c_H > 1 - 1/(1+c_T)$. Thus, we have developed an analogue of AM-IHT that works purely with the RIP-1 assumption on the
measurement matrix and hence is suitable for recovery using
sparse matrices. It is likely that a similar analogue can be developed
for AM-CoSaMP, but we will not pursue this direction here.

\section{Improved Recovery via Boosting}
\label{sec:boosting}
As stated in Sections \ref{sec:amiht} and \ref{sec:amcosamp}, AM-IHT and AM-CoSaMP require stringent assumptions on the head- and tail-approximation factors $c_H$ and $c_T$.
The condition \eqref{eq:headtail_iht} indicates that for AM-IHT to converge, the head- and tail-approximation factors must be tightly coupled. 
Observe that by definition, $c_T$ is no smaller than $1$.
Therefore, $c_H$ must be at least $\sqrt{3}/2$.
If $c_T$ is large (i.e., if the tail-approximation oracle gives only a crude approximation), then the head-approximation oracle needs to be even more precise.
For example, if $c_T = 10$, then $c_H > 0.995$, i.e., the head approximation oracle needs to be very accurate.
Such a stringent condition can severely constrain the choice of approximation algorithms.

In this section, we overcome this barrier by demonstrating how to ``boost'' the approximation factor of any given head-approximation algorithm.
Given a head-approximation algorithm with arbitrary approximation factor $c_H$, we can boost its approximation factor to any arbitrary constant $c_H' < 1$.
Our approach requires only a constant number of invocations of the original head-approximation algorithm and inflates the sample complexity of the resulting output model only by a constant factor.
Combining this boosted head-approximation algorithm with AM-IHT or AM-CoSaMP, we can provide an overall recovery scheme for approximation algorithms with \emph{arbitrary} approximation constants $c_T$ and $c_H$.
This is a much weaker condition than \eqref{eq:headtail_iht} and therefore significantly extends the scope of our framework for model-based compressive sensing with approximate projection oracles.

We achieve this improvement by iteratively applying the head-approximation algorithm to the residual of the currently selected support.
Each iteration guarantees that we add another $c_H$-fraction of the best remaining support to our result.
Algorithm \ref{alg:boosting} contains the corresponding pseudo code and Theorem \ref{thm:boosting} the main guarantees.

\begin{algorithm}[!t]
\caption{Boosting for head-approximation algorithms}
\label{alg:boosting}
\begin{algorithmic}[1]
\Function{BoostHead}{$x,H,t$}
\State $\Omega_0 \gets \{ \}$
\For{$i \gets 1, \ldots, t$}
  \State $\Lambda_i \gets H(x_{[n] \setminus \Omega_{i - 1}})$ \label{line:usehead}
  \State $\Omega_{i} \gets \Omega_{i-1} \cup \Lambda_i$
\EndFor
\State \textbf{return} $\Omega_{t}$
\EndFunction
\end{algorithmic}
\end{algorithm}

\begin{theorem}
\label{thm:boosting}
Let $H$ be a $(c_H, \Msupports, \Msupports_H, p)$-head-approximation algorithm with $0 < c_H \leq 1$ and $p \geq 1$.
Then \textsc{BoostHead}$(x,H,t)$ is a $((1 - (1 - c^p_H)^t)^{1/p}, \Msupports, \Msupports_H^{\oplus t}, p)$-head-approximation algorithm.
Moreover, \textsc{BoostHead} runs in time $O(t \cdot T_H)$, where $T_H$ is the time complexity of $H$.
\end{theorem}
\begin{proof}
Let $\Gamma \in \Msupports$ be an optimal support, i.e., $\norm{x_\Gamma}_p = \max_{\Omega \in \Msupports} \norm{x_\Omega}_p$.
We now prove that the following invariant holds at the beginning of iteration $i$:
\begin{equation}
\label{eq:invariant}
  \norm{x_\Gamma}_p^p - \norm{x_{\Omega_{i-1}}}_p^p \leq  (1 - c_H^p)^{i-1} \norm{x_\Gamma}_p^p \; .
\end{equation}
Note that the invariant (Equation \ref{eq:invariant}) is equivalent to $\norm{x_{\Omega_{i-1}}}_p^p \geq \parens{1 - (1-c_H^p)^{i-1}} \norm{x_\Gamma}_p^p$.
For $i = t+1$, this gives the head-approximation guarantee stated in the theorem.

For $i=1$, the invariant directly follows from the initialization.

Now assume that the invariant holds for an arbitrary $i \geq 1$.
From line \ref{line:usehead} we have
\begin{align*}
  \norm{(x_{[n] \setminus \Omega_{i-1}})_{\Lambda_i}}_p^p &\geq c_H^p \max_{\Omega \in \Msupports} \norm{(x_{[n] \setminus \Omega_{i-1}})_\Omega}_p^p \\
  \norm{x_{\Lambda_i \setminus \Omega_{i-1}}}_p^p &\geq c_H^p \max_{\Omega \in \Msupports} \norm{(x - x_{\Omega_{i-1}})_\Omega}_p^p \\
            &\geq c_H^p \norm{(x - x_{\Omega_{i-1}})_\Gamma}_p^p \\
            &= c_H^p \norm{x_\Gamma - x_{\Omega_{i-1} \cap \Gamma}}_p^p \\
            &= c_H^p \parens{\norm{x_\Gamma}_p^p - \norm{x_{\Omega_{i-1} \cap \Gamma}}_p^p} \\
            &\geq c_H^p \parens{\norm{x_\Gamma}_p^p - \norm{x_{\Omega_{i-1}}}_p^p} \; . \numberthis \label{eq:boostingmax}
\end{align*}

We now prove the invariant for $i+1$:
\begin{align*}
  \norm{x_\Gamma}_p^p - \norm{x_{\Omega_i}}_p^p &= \norm{x_\Gamma}_p^p - \norm{x_{\Omega_{i-1}}}_p^p - \norm{x_{\Lambda_i \setminus \Omega_{i-1}}}_p^p \\
        &\leq \norm{x_\Gamma}_p^p - \norm{x_{\Omega_{i-1}}}_p^p - c_H^p \parens{\norm{x_\Gamma}_p^p - \norm{x_{\Omega_{i - 1}}}_p^p} \\
        &= (1 - c_H^p) \parens{\norm{x_\Gamma}_p^p - \norm{x_{\Omega_{i - 1}}}_p^p} \\
        &\leq (1 - c_H^p)^{i+1} \norm{x_\Gamma}_p^p \; .
\end{align*}
The second line follows from \eqref{eq:boostingmax} and the third line from the invariant.

Since $\Lambda_i \in \Msupports_H$, we have $\Omega_t \in \Msupports_H^{\oplus t}$.
The time complexity of \textsc{BoostHead} follows directly from the definition of the algorithm.
\end{proof}

We now use Theorem \ref{thm:boosting} to relax the conditions on $c_T$ and $c_H$ in Corollary \ref{cor:amiht}.
As before, we assume that we have compressive measurements of the form $y = Ax + e$, where $x \in \Mmodel$ and $e$ is arbitrary measurement noise.

\begin{corollary}
\label{cor:boosting}
Let $T$ and $H$ be approximate projection oracles with $c_T \geq 1$ and $0 < c_H < 1$.
Moreover, let $\delta$ be the model-RIP constant of the measurement matrix $A$ and let
\begin{align*}
  \gamma &= \frac{\sqrt{1 - \left( \frac{1}{1 + c_T} - \delta \right)^2} + \delta}{1 - \delta} \, ,\\
  t &= \ceil{\frac{\log (1-\gamma^2)}{\log (1-c^2_H)}}  + 1\, .
\end{align*}
We assume that $\delta$ is small enough so that $\gamma < 1$ and that $A$ satisfies the model-RIP for $\Msupports \supportplus \Msupports_T \supportplus \Msupports_H^{\supportplus t}$.
Then AM-IHT with $T$ and \textsc{BoostHead}($x, H, t$) as projection oracles returns a signal estimate $\widehat{x}$ satisfying
\[
  \norm{x - \widehat{x}}_2 \leq C \norm{e}_2 
\]
after $O(\log{\frac{\norm{x}_2}{\norm{e}_2}})$ iterations.
The constants in the error and runtime bounds depend only on $c_T$, $c_H$, and $\delta$.
\end{corollary}
\begin{proof}
In order to use Corollary \ref{cor:amiht}, we need to show that $\alpha < 1$.
Recall that
\[
  \alpha = (1 + c_T) (\delta + \sqrt{1 - (c_H (1-\delta) - \delta)^2}) \, .
\]
A simple calculation shows that a head-approximation oracle with $c'_H > \gamma$ achieves $\alpha < 1$.

Theorem \ref{thm:boosting} shows that boosting the head-approximation oracle $H$ with $t'$ iterations gives a head-approximation factor of
\[
  c'_H = \sqrt{1 - (1-c^2_H)^{t'}} \, .
\]
Setting $t' = t$ as defined in the theorem yields $c'_H > \gamma$.
We can now invoke Corollary \ref{cor:amiht} for the recovery guarantee of AM-IHT.
\end{proof}

Analogous corollaries can be proven for AM-CoSaMP (Section \ref{sec:amcosamp}) and AM-IHT with RIP-1 (Section \ref{sec:modelrip1}).
We omit detailed statements of these results here.

\section{Case Study: The CEMD model}
\label{s:emd}
As an instantiation of our main results, we discuss a special structured sparsity model known as the \emph{Constrained EMD} model~\cite{HIS}. 
A key ingredient in the model is the Earth Mover's Distance (EMD), also known as the Wasserstein metric or Mallows distance~\cite{emdismallows}:
\begin{definition}[EMD]\label{def:emd}
The EMD of two finite sets $A, B \subset \N$ with $|A| = |B|$ is defined as
\begin{equation}
\EMD(A,B) = \min_{\pi: A \rightarrow B} \sum_{a \in A} \abs{a - \pi(a)} \; ,
\end{equation}
where $\pi$ ranges over all one-to-one mappings from $A$ to $B$.
\end{definition}
Observe that $\EMD(A,B)$ is equal to the cost of a min-cost matching between $A$ and $B$.
Now, consider the case where the sets $A$ and $B$ are the {\em supports} of two exactly $k$-sparse signals, so that $|A| = |B| = k$.
In this case, the EMD not only measures how many indices change, but also how far the supported indices move.
This notion can be generalized from pairs of signals to an \emph{ensemble} of sparse signals.
Figure \ref{fig:supportemd} illustrates the following definition.

\begin{figure}[t!]
\centering
\begin{tikzpicture}[scale=0.4]
\tikzstyle{line}=[very thick,-,shorten <=2pt,shorten >=2pt]
\tikzstyle{support}=[fill=black,draw=black];

\node [anchor=south] (xilabel) at (0.5,8) {$X_{*,1}$};
\draw[thick] (0,0) grid +(1,8);
\draw[support] (0.5,5.5) circle (.3);
\draw[support] (0.5,2.5) circle (.3);
\draw[support] (0.5,0.5) circle (.3);

\node [anchor=south] (xip1label) at (5.5,8) {$X_{*,2}$};
\draw[thick] (5,0) grid +(1,8);
\draw[support] (5.5,7.5) circle (.3);
\draw[support] (5.5,1.5) circle (.3);
\draw[support] (5.5,0.5) circle (.3);

\draw[line] (1,5.5) -- node [above] {2} (5,7.5);
\draw[line] (1,2.5) -- node [above] {1} (5,1.5);
\draw[line] (1,0.5) -- node [above] {0} (5,0.5);
\node [anchor=north] (semdlabel1) at (3,-0.5) {$\EMD=3$};

\node [anchor=south] (xip2label) at (10.5,8) {$X_{*,3}$};
\draw[thick] (10,0) grid +(1,8);
\draw[support] (10.5,7.5) circle (.3);
\draw[support] (10.5,2.5) circle (.3);
\draw[support] (10.5,1.5) circle (.3);

\draw[line] (6,7.5) -- node [above] {0} (10,7.5);
\draw[line] (6,1.5) -- node [above] {1} (10,2.5);
\draw[line] (6,0.5) -- node [below] {1} (10,1.5);
\node [anchor=north] (semdlabel2) at (8,-0.5) {$\EMD=2$};
\end{tikzpicture}
\caption[The support-EMD]{The support-EMD for a matrix with three columns and eight rows.
The circles stand for supported elements in the columns.
The lines indicate the matching between the supported elements and the corresponding EMD cost.
The total support-EMD is $\EMD(\supp(X)) = 2 + 3 = 5$.}
\label{fig:supportemd}
\end{figure}
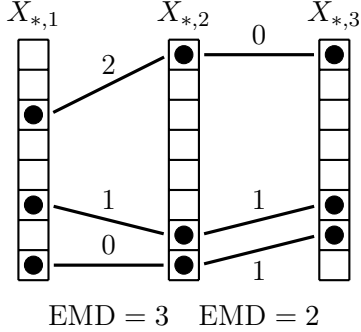

\begin{definition}[Support-EMD]
Let $\Omega \subseteq [h] \times [w]$ be the support of a matrix $X$ with exactly $s$-sparse columns, i.e., $|\colsupp(\Omega,c)| = s$ for $c \in [w]$.
Then the EMD of $\Omega$ is defined as
$$
\EMD(\Omega) = \sum_{c=1}^{w-1} \EMD(\colsupp(\Omega,c), \colsupp(\Omega,c+1)) \; .
$$

If the columns of $X$ are not exactly $s$-sparse, we define the EMD of $\Omega$ as the minimum EMD of any support that contains $\Omega$ and has exactly $s$-sparse columns.
Let $s = \max_{c\in [w]} |\colsupp(\Omega,c)|$.
Then $\EMD(\Omega) = \min_\Gamma \EMD(\Gamma)$, where $\Gamma \subseteq [h] \times [w]$, $\Omega \subseteq \Gamma$, and $\Gamma$ is a support with exactly $s$-sparse columns, i.e.,\ $| \colsupp(\Gamma,c) | = s$ for $c \in [w]$.
\end{definition}
The above definitions motivate a natural structured sparsity model that essentially characterizes ensembles of sparse signals with correlated supports.
Suppose we interpret the signal $x \in \R^n$ as a matrix $X \in \R^{h \times w}$ with $n = h \, w$.
For given dimensions of the signal $X$, our model has two parameters:
(i) $k$, the total sparsity of the signal.
For simplicity, we assume here and in the rest of this paper that $k$ is divisible by $w$.
Then the sparsity of each column $X_{*,i}$ is $s = k / w$.
(ii) $B$, the support-EMD of $X$. We call this parameter the \emph{EMD budget}.
Formally, we have:

\begin{definition}[Constrained EMD model]
\label{def:emdmodel}
The Constrained EMD (CEMD) model is the structured sparsity model $\Mmodel_{k,B}$ defined by the set of supports $\Msupports_{k,B} = \{ \Omega \subseteq [h] \times [w] \, | \, \EMD(\Omega) \leq B  \textnormal{ and } |\colsupp(\Omega, c)| = \frac{k}{w} \textnormal{ for } c \in [w] \}$.
\end{definition}

The parameter $B$ controls how much the support can vary from one column to the next.
Setting $B=0$ forces the support to remain constant across all columns, which corresponds to block sparsity (the blocks are the rows of $X$).
A value of $B \geq kh$ effectively removes the EMD constraint because each supported element is allowed to move across the full height of the signal.
In this case, the model demands only $s$-sparsity in each column.
It is important to note that we only constrain the EMD of the column \emph{supports} in the signal, not the actual amplitudes.
Figure \ref{fig:emdmodel_example} illustrates the CEMD model with an example.

\begin{figure}[t!]
\centering
\begin{tikzpicture}
\tikzstyle{amp}=[minimum size=.6cm,node distance=0.8cm,inner sep=2pt]
\tikzstyle{line}=[-,thick]
\tikzstyle{bracket}=[thick]

\node (n00) at (0,0) [amp] {1};
\node (n10) [amp, below=of n00] {0};
\node (n20) [amp, below=of n10] {4};
\node (n01) [amp, right=of n00] {3};
\node (n11) [amp, below=of n01] {1};
\node (n21) [amp, below=of n11] {2};
\node (n02) [amp, right=of n01] {1};
\node (n12) [amp, below=of n02] {2};
\node (n22) [amp, below=of n12] {0};
\draw [bracket] (n00.north) -- (n00.north west) -- (n20.south west) -- (n20.south);
\draw [bracket] (n02.north) -- (n02.north east) -- (n22.south east) -- (n22.south);
\node (label) [left=.1cm of n10] {$X =$};

\node (m00) [amp, right=4cm of n02] {0};
\node (m10) [amp, below=of m00] {0};
\node (m20) [amp, below=of m10] {4};
\node (m01) [amp, right=of m00] {0};
\node (m11) [amp, below=of m01] {0};
\node (m21) [amp, below=of m11] {2};
\node (m02) [amp, right=of m01] {0};
\node (m12) [amp, below=of m02] {2};
\node (m22) [amp, below=of m12] {0};
\draw [bracket] (m00.north) -- (m00.north west) -- (m20.south west) -- (m20.south);
\draw [bracket] (m02.north) -- (m02.north east) -- (m22.south east) -- (m22.south);
\node (label) [left=.1cm of m10] {$X^* =$};
\draw [line] (m20) -- (m21) node [midway,below,font=\scriptsize] {0};
\draw [line] (m21) -- (m12) node [midway,below right,font=\scriptsize] {1};

\end{tikzpicture}
\caption[The EMD model]{A signal $X$ and its best approximation $X^*$ in the EMD model $\Mmodel_{3,1}$.
A sparsity constraint of 3 with 3 columns implies that each column has to be 1-sparse.
Moreover, the total support-EMD between neighboring columns in $X^*$ is 1.
The lines in $X^*$ indicate the support-EMD.}
\label{fig:emdmodel_example}
\end{figure}

\subsection{Sampling bound}\label{sec:samplingbound}
Our objective is to develop a sparse recovery scheme for the Constrained EMD model.
As the first ingredient, we establish the model-RIP  for $\Mmodel_{k,B}$, i.e., we characterize the number of permissible supports (or equivalently, the number of subspaces) $l_{k,B}$ in the model and invoke Fact~\ref{fact:modelripbound}. 
For simplicity, we will assume that $w = \Omega(\log h)$, i.e., the following bounds apply for all signals $X$ except very thin and tall matrices $X$.
The following result is novel:

\begin{theorem}
\label{thm:samplingbound}
The number of allowed supports in the CEMD model satisfies
$\log \abs{\Msupports_{k,B}} = O\parens{k \log \frac{B}{k}}.$
\end{theorem}
\begin{proof}
For given $h$, $w$, $B$, and $k$, the support is fixed by the following three decisions:
(i) The choice of the supported elements in the first column of $X$.
(ii) The distribution of the EMD budget $B$ over the $k$ supported elements.
This corresponds to distributing $B$ balls into $k + 1$ bins (using one bin for the part of the EMD budget not allocated to supported elements).
(iii) For each supported element, the direction (up or down) to the matching element in the next column to the right.
Multiplying the choices above gives $\binom{h}{s} \binom{B + k}{k} 2^k$, an upper bound on the number of supports.
Using the inequality $\binom{a}{b} \leq \parens{\frac{a \, e}{b}}^b$, we get
\begin{align*}
\log \abs{\Msupports_{k,B}} & \leq \log \parens{\binom{h}{s} \binom{B + k}{k} 2^k} \\
             & \leq s \log \frac{h}{s} + k \log \frac{B + k}{k} + O(s + k) \\
             &= O\parens{k \log \frac{B}{k}} \, . \, \qedhere
\end{align*}
\end{proof}

If we allow each supported element to move a constant amount from one
column to the next, we get $B = O(k)$ and hence, from
Fact~\ref{fact:modelripbound},  $m = O(k + \log \abs{\Msupports_{k,B}}) = O(k)$ rows for sub-Gaussian measurement matrices.
This bound is information-theoretically optimal.
Furthermore, for $B = kh$ (i.e., allowing every supported element to move anywhere in the next column) we get $m = O(k \log n)$, which almost matches the standard compressive sensing bound of $m = O(k \log \frac{n}{k})$ for sub-Gaussian measurement matrices.
Therefore, the CEMD model gives a smooth trade-off between the support variability and the number of measurements necessary for recovery.

We can also establish a sampling bound in the RIP-1 setting with Fact \ref{fact:modelripbound1}.
For the case of $B = \Theta(k)$, we get $m = O(k \, \frac{\log n}{\log \log \frac{n}{k}})$.
In order to match the block-sparsity lower bound of $m = O(k \log_w n)$, we need to assume that $B = O(k / w)$, i.e., each path (and not each element) in the support has a constant EMD-budget on average.
We omit the details of this calculation here.

The following theorem is useful when establishing sampling bounds for recovery schemes using the CEMD model.
\begin{theorem}
\label{thm:cemdaddition}
The CEMD model is closed under addition:
$\Msupports_{k_1, B_1} \oplus \Msupports_{k_2, B_2} \subseteq \Msupportsclosure_{k_1 + k_2, B_1 + B_2}$.
\end{theorem}
\begin{proof}
Let $\Omega_1 \in \Msupports_{k_1, B_1}$ and $\Omega_2 \in \Msupports_{k_2, B_2}$.
Moreover, let $\Gamma = \Omega_1 \cup \Omega_2$.
We have to show that $\Gamma \in \Msupports_{k_1+k_2, B_1 + B_2}$.

The column-sparsity of $\Omega_1$ and $\Omega_2$ is $k_1/w$ and $k_2/w$, respectively.
Hence the column-sparsity of $\Gamma$ is at most $\frac{k_1 + k_2}{w}$.
Moreover, we can construct a matching for $\Gamma$ with cost at most $B_1 + B_2$ from the matchings for $\Omega_1$ and $\Omega_2$.
To see this, consider without loss of generality the matchings $\pi_1$ and $\pi_2$ corresponding to the first two columns in $\Omega_1$ and $\Omega_2$, respectively.
We start constructing the new matching $\pi'$ by starting with $\pi_1$.
Then, we iterate over the pairs $(a, b)$ in $\pi_2$ one by one and augment $\pi'$ to include both $a$ and $b$.
There are four cases:
\begin{enumerate}
\item Both $a$ and $b$ are still unassigned in $\pi'$.
Then we can simply add $(a, b)$ to $\pi'$.
\item Both $a$ and $b$ are already assigned in $\pi'$.
In this case, we do not need to modify $\pi'$ to include $a$ and $b$.
\item $a$ is not included in $\pi'$, but $b$ is already assigned in $\pi'$.
This is the interesting case becaues we must now find a new neighbor assignment for $a$.
Let $b'$ be the entry in the second column that is in the same row as $a$.
If $b'$ is not assigned yet, we can simply add $(a, b')$ to $\pi'$.
Otherwise, let $a'$ be the value such that $\pi'(a') = b'$.
Then we remove the pair $(a', b')$ from $\pi'$, add $(a, b')$ to $\pi'$, and repeat this procedure to find a new neighbor for $a'$.
It is easy to see that this procedure terminates after a finite number of steps, and that no node currently assigned under $\pi'$ loses a neighbor.
Moreover, note that this operation does not increase the cost of the matching $\pi'$.
\item $b$ is not included in $\pi'$, but $a$ is already assigned in $\pi'$.
This case is symmetric to case 3 above.
\end{enumerate}
Each of the four cases increases the cost of $\pi'$ by at most the cost of $(a, b)$ in $\pi_2$. Iterating over all pairs in $\pi_2$, we observe that the final matching $\pi'$ has cost no more than the cumulative costs of $\pi_1$ and $\pi_2$, i.e., at most $B_1 + B_2$.
Therefore, $\Gamma \in \Msupports_{k_1 + k_2, B_1 + B_2}$.
\end{proof}

\subsection{Head Approximation Algorithm}
\label{sec:head}
First, we develop a head approximation algorithm for the CEMD model.
Ideally, we would have an \emph{exact} projection algorithm $H$ mapping arbitrary signals to signals in $\Mmodel_{k,B}$ with the guarantee $\norm{H(x)}_p = \max_{\Omega \in \Msupports_{k,B}} \norm{x_\Omega}_p$.
However, this appears to be a hard problem. Instead, we propose an
efficient greedy algorithm satisfying the somewhat looser requirements of a head
approximation oracle
(Definition~\ref{def:headapproxoracle}). 
Specifically, we develop an algorithm that performs the following task: 
given an arbitrary signal $x$, find a support $\Omega \in \Msupports_{O(k),O(B \log k)}$ such that
$\norm{x_\Omega}_p^p \geq c \max_{\Gamma\in \Msupports_{k,B}} \norm{x_\Gamma}_p^p$, where $c > 0$ is a fixed constant.

As before, we interpret our signal $x$ as a matrix $X \in \R^{h\times w}$.
Let $\OPT$ denote the largest sum of coefficients achievable with a support in $\Msupports_{k,B}$, i.e., $\OPT = \max_{\Omega \in \Msupports_{k,B}} \norm{x_\Omega}_p^p$.
For a signal $x \in \Mmodel_{k,B}$, we interpret the support of $x$ as a set of $s = k/w$ paths from the leftmost to the rightmost column in $X$.
Our method proceeds by greedily finding a set of paths that cover a large sum of signal coefficients.
We can then show that the coefficients covered by these paths are a constant fraction of the optimal coefficient sum $\OPT$.

\begin{definition}[Path in a matrix]\label{def:matrixpath}
Given a matrix $X \in \R^{h\times w}$, a path $r \subseteq [h] \times [w]$ is a set of $w$ locations in $X$ with one location per column, i.e., $\abs{r} = w$ and $\bigcup_{(i,j) \in r} j = [w]$.
The weight of $r$ is the sum of amplitudes on $r$, i.e., $w_{X,p}(r) = \sum_{(i,j)\in r} \abs*{X_{i,j}}^p \, $.
The EMD of $r$ is the sum of the EMDs between locations in neighboring columns.
Let $j_1, \ldots, j_w$ be the locations of $r$ in columns 1 to $w$.
Then, 
$\EMD(r) = \sum_{i=1}^{w-1} \abs{j_i - j_{i+1}} \, .$
\end{definition}

\begin{algorithm}[!t]
\caption{Head approximation algorithm}
\label{alg:headapproxbasic}
\begin{algorithmic}[1]
\Function{HeadApprox}{$X, k, B$}
\State $X^{(1)} \gets X$
\For{$i  \gets 1, \ldots, s$}
\State Find the path $r_i$ from column 1 to column $w$ in $X^{(i)}$ that maximizes $w^{(i)}(r_i)$ and
\Statex uses at most EMD-budget $\floor{\frac{B}{i}}$.
\State $X^{(i + 1)} \gets X^{(i)}$
\For{$(u,v)  \in r_i $}
\State $X^{(i+1)}_{u,v} \gets 0$
\EndFor
\EndFor
\State \textbf{return} $\bigcup_{i = 1}^{s} r_i$
\EndFunction
\end{algorithmic}
\end{algorithm}

Trivially, we have that a path $r$ in $X$ is a support with $w_{X,p}(r) = \norm{X_r}_p^p$ and $\EMD(r) = \EMD(\supp(X_r))$.
Therefore, we can iteratively build a support $\Omega$ by finding $s$ paths in $X$.
Algorithm~\ref{alg:headapproxbasic} contains the description of \textsc{HeadApprox}.
We show that \textsc{HeadApprox} finds a constant fraction of the amplitude sum of the best support while only moderately increasing the size of the model.
For simplicity, denote $w(r) := w_{X,p}(r)$, and $w^{(i)}(r) := w_{X^{(i)},p}(r)$.
We obtain the following result:

\begin{theorem}
\label{thm:headapproxbasic}
Let $p \geq 1$ and $B' = \ceil{H_s} B$, where $H_s = \sum_{i=1}^s 1/i$ is the $s$-th harmonic number.
Then \textsc{HeadApprox} is a $((\frac{1}{4})^{1/p}, \Msupports_{k,B}, \Msupports_{k, B'}, p)$-head-approximation oracle.
\end{theorem}
\begin{proof}
Let $\Omega$ be the support returned by $\textsc{HeadApprox}(X, k, B)$ and let $\Omega_\OPT \in \Msupports_{k,B}$ be an optimal support.
We can always decompose $\Omega_\OPT$ into $s$ disjoint paths in $X$.
Let $t_1, \ldots, t_s$ be such a decomposition with $\EMD(t_1) \geq \EMD(t_2) \geq \ldots \geq \EMD(t_s)$.
Note that $\EMD(t_i) \leq \floor{\frac{B}{i}}$: otherwise $\sum_{j=1}^i \EMD(t_i) > B$ and since $\EMD(\Omega_\OPT) \leq B$ this would be a contradiction.
Since $\Omega$ is the union of $s$ disjoint paths in $X$, $\Omega$ has column-sparsity $s$.
Moreover, we have $\EMD(\Omega) = \sum_{i=1}^s \EMD(r_i) \leq \sum_{i=1}^s \floor{\frac{B}{i}} \leq \ceil{H_s} B .$
Therefore, $\Omega \in \Msupports_{k,B'}^+$.

When finding path $r_i$ in $X^{(i)}$, there are two cases:
\begin{description}[font=\normalfont]
\item[Case 1:] $w^{(i)}(t_i) \leq \frac{1}{2} w(t_i)$, i.e., the paths $r_1, \ldots, r_{i-1}$ have already covered more than half of the coefficient sum of $t_i$ in $X$.
\item[Case 2:] $w^{(i)}(t_i) > \frac{1}{2} w(t_i)$, i.e., there is still more than half of the coefficient sum of $t_i$ remaining in $X^{(i)}$.
Since $\EMD(t_i) \leq \floor{\frac{B}{i}}$, the path $t_i$ is a candidate when searching for the optimal path $r_i$ and hence we find a path $r_i$ with $w^{(i)}(r_i) > \frac{1}{2} w(t_i)$.
\end{description}
Let $C = \{i \in [s] \; | \; \text{case 1 holds for } r_i\}$ and $D = \{i \in [s] \; | \; \text{case 2 holds for } r_i\}$ (note that $C = [s] \setminus D$).
Then we have
\begin{equation}
\begin{split}
\norm{X_\Omega}_p^p = \sum_{i=1}^s w^{(i)}(r_i) 
& = \sum_{i \in C} w^{(i)}(r_i) + \sum_{i \in D} w^{(i)}(r_i) \\
& \geq \sum_{i \in D} w^{(i)}(r_i) 
\geq \frac{1}{2} \sum_{i \in D} w(t_i) \label{eq:dsum} \, .
\end{split}
\end{equation}
For each $t_i$ with $i \in C$, let $E_i = t_i \cap \, \bigcup_{j < i} r_j$, i.e., the locations of $t_i$ already covered by some $r_j$ when searching for $r_i$.
Then we have
\begin{equation*}
\sum_{(u,v) \in E_i} |X_{u,v}|^p = w(t_i) - w^{(i)}(t_i) \geq  \frac{1}{2} w(t_i) \, ,
\end{equation*}
and
\begin{equation*}
\sum_{i \in C}\sum_{(u,v) \in E_i} |X_{u,v}|^p \geq  \frac{1}{2} \sum_{i \in C} w(t_i) \, .
\end{equation*}
The $t_i$ are pairwise disjoint, and so are the $E_i$.
For every $i \in C$ we have $E_i \subseteq \bigcup_{j=1}^s r_j$. Hence
\begin{equation}
\norm{X_\Omega}_p^p = \sum_{i=1}^s w^{(i)}(r_i) 
\geq \sum_{i \in C}\sum_{(u,v) \in E_i} |X_{u,v}|^p 
\geq  \frac{1}{2} \sum_{i \in C}  w(t_i) \,.
\label{eq:csum}
\end{equation}
Combining Equations \ref{eq:dsum} and \ref{eq:csum} gives:
\begin{align*}
2 \norm{X_\Omega}_p^p \; \geq & \; \frac{1}{2} \sum_{i \in C}  w(t_i) + \frac{1}{2} \sum_{i \in D} w(t_i)  = \frac{1}{2} \OPT \\
\norm{X_\Omega}_p \; \geq & \; \parens{\frac{1}{4}}^{1/p} \max_{\Omega' \in \Msupports_{k,B}} \norm{X_{\Omega'}}_p \; .
\end{align*}
\end{proof}

\begin{theorem}
\label{thm:headapproxtime}
\textsc{HeadApprox} runs in $O(s  n  B  h)$ time.
\end{theorem}
\begin{proof}
Observe that the running time of \textsc{HeadApprox} depends on the running time of finding a path with maximum weight for a given EMD budget.
The search for such a path can be performed by \emph{dynamic programming} over a graph with $w h B = n B$ nodes, or equivalently ``states'' of the dynamic program.\footnote{We use the terminology ``states'' here to distinguish the dynamic program from the graph we will introduce in Section \ref{sec:tail}.}
Each state in the graph corresponds to a state in the dynamic program, i.e., a location $(i, j) \in [w] \times [h]$ and the current amount of EMD already used $b \in \{0, 1, \ldots, B\}$.
At each state, we store the largest weight achieved by a path ending at the corresponding location $(i,j)$ and using the corresponding amount of EMD budget $b$.
Each state has $h$ outgoing edges to the states in the next column (given the current location, the decision on the next location also fixes the new EMD amount).
Hence the time complexity of finding one largest-weight path is $O(n B h)$ (the state space has size $O(n B)$ and each update requires $O(h)$ time).
Since we repeat this procedure $s$ times, the overall time complexity of \textsc{HeadApprox} is $O(s n B h)$.
\end{proof}

We can achieve an arbitrary constant head-approximation ratio by combining \textsc{HeadApprox} with \textsc{BoostHead} (see Section \ref{sec:boosting}).
The resulting algorithm has the same time complexity as \textsc{HeadApprox}.
Moreover, the sparsity and EMD budget of the resulting support is only a constant factor larger than $k$ and $B'$.

\subsection{Tail-Approximation Algorithm}
\label{sec:tail}
Next, we develop a tail-approximation algorithm for the CEMD model.
Given an arbitrary signal $x$, our objective is to find a support $\Gamma \in \Msupports_{k, O(B)}$ such that
\begin{equation}
\label{eq:tailapprox}
  \norm{x - x_\Gamma}_p \leq c \min_{\Omega\in \Msupports_{k,B}} \norm{x - x_\Omega}_p \, ,
\end{equation}
where $c$ is a constant.
Note that we allow a constant factor increase in the EMD budget of the result.
The algorithm we develop is precisely the graph-based approach initially proposed in~\cite{HIS}; however, our analysis here is rigorous and novel.
Two core elements of the algorithm are the notions of a \emph{flow network} and the \emph{min-cost max-flow problem}, which we now briefly review.
We refer the reader to \cite{AMO93} for an introduction to the graph-theoretic definitions and algorithms we employ.

The min-cost max-flow problem is a generalization of the classical maximum flow problem \cite{CLRS01,AMO93}.
In this problem, the input is a graph $G = (V,E)$ with designated source and sink nodes in which every edge has a certain capacity.
The goal is to find an assignment of flow to edges such that the total flow from source to sink is maximized.
The flow must also be valid, i.e., the amount of flow entering any intermediate node must be equal to the amount of flow leaving that intermediate node, and the amount of flow on any edge can be at most the capacity of that edge.

In the min-cost max-flow problem, every edge $e$ also has a cost $c_e$ (in addition to the capacity as before).
The goal now is to find a flow $f: E \rightarrow \R^+_0$ with maximum capacity such that the cost of the flow, i.e., $\sum_{e \in E} c_e \cdot f(e)$, is minimized.
One important property of the min-cost max-flow problem is that it still admits \emph{integral} solutions if the edge capacities are integer.

\begin{fact}[Theorem 9.10 in \cite{AMO93}]
If all edge capacities, the source supply, and the sink demand are integers, then there is always an integer min-cost max-flow.
\end{fact}

The min-cost max-flow problem has many applications, and several efficient algorithms are known \cite{AMO93}.
We leverage this problem for our tail-approximation task by carefully constructing a suitable flow network, which we now define.

\begin{definition}[EMD flow network]
For a given signal $X$, sparsity $k$, and a parameter $\lambda
> 0$, the flow network $G_{X,k,\lambda}$ consists of the following elements:
\begin{itemize}[nosep]
\item The \emph{nodes} comprise a source, a sink and a node $v_{i,j}$ for $i \in [h]$, $j \in [w]$, i.e., one node per entry in $X$ (besides source and sink).
\item $G$ has an \emph{edge} from every $v_{i,j}$ to every $v_{k,j+1}$ for \mbox{$i, k \in [h]$}, $j \in [w-1]$.
Moreover, there is an edge from the source to every $v_{i,1}$ and from every $v_{i,w}$ to the sink.
\item The \emph{capacity} on every edge and node (except source and sink) is 1.
\item The \emph{cost} of node $v_{i,j}$ is $-\abs{X_{i,j}}^p$.
The cost of an edge from $v_{i,j}$ to $v_{k,j+1}$ is $\lambda |i - k|$.
The cost of the source, the sink, and all edges incident to the source
or sink is 0.
\item The \emph{supply} at the source is $s$ ($= \frac{k}{w}$) and the demand at the sink is $s$.
\end{itemize}
\end{definition}
Figure~\ref{fig:emdflownetwork_example} illustrates this definition with an example.
The main idea is that a set of disjoint paths through the network $G_{X,k,\lambda}$ corresponds to a support in $X$.
For any fixed value of $\lambda$, a solution of the min-cost max-flow problem on the flow network reveals a subset $S$ of the nodes that corresponds to a support with exactly $s$ indices per column and minimizes  $-\norm{X_\Omega}_p^p + \lambda \EMD(\Omega)$ for different choices of support $\Omega$.
In other words, the min-cost flow solves  a \emph{Lagrangian relaxation} of the original problem \eqref{eq:tailapprox}.
See Lemmas \ref{thm:emdflownetwork_correspondence} and \ref{lemma:emdflow} for a more formal statement of this connection.

\begin{figure}[ht!]
\centering
\begin{tikzpicture}
\tikzstyle{graph}=[circle,minimum size=.6cm,draw,thick,node distance=1.2cm,inner sep=1pt]
\tikzstyle{edge}=[->,thick]
\node (v21) at (0,0) [graph] {0};
\node (v11) [graph, above=of v21] {-1};
\node (v31) [graph, below=of v21] {-2};
\node (v22) [graph, right=of v21] {-1};
\node (v12) [graph, above=of v22] {-3};
\node (v32) [graph, below=of v22] {-1};
\node (s) [graph, left=of v21] {};
\node (t) [graph, right=of v22] {};
\node (slabel) [below left=0cm of s,inner sep=2pt] {source};
\node (tlabel) [below right=0cm of t,inner sep=1pt] {sink};
\draw [edge] (s) -- (v11);
\draw [edge] (s) -- (v21);
\draw [edge] (s) -- (v31);
\draw [edge] (v11) -- (v12) node [midway,above] {$0$};
\draw [edge] (v11) -- (v22) node [midway,above,pos=0.4] {$\lambda$};
\draw [edge] (v11) -- (v32) node [midway,below,anchor=east,pos=0.22,inner sep=1pt] {$\,2\lambda$};
\draw [edge] (v21) -- (v12);
\draw [edge] (v21) -- (v22);
\draw [edge] (v21) -- (v32);
\draw [edge] (v31) -- (v12);
\draw [edge] (v31) -- (v22);
\draw [edge] (v31) -- (v32);
\draw [edge] (v12) -- (t);
\draw [edge] (v22) -- (t);
\draw [edge] (v32) -- (t);

\node (graphlabel) [above left=0.8cm of s] {$G_{X,k,\lambda} = $};

\node (signal) [left=4.8cm of s]{
$X = 
\begin{bmatrix*}[r]
1 & 3 \\[.6cm]
0 & -1 \\[.6cm]
2 & 1 \\
\end{bmatrix*}$
};
\end{tikzpicture}
\caption[Example of an EMD flow network]{A signal $X$ with the corresponding flow network $G_{X,k,\lambda}$ for $p=1$.
The node costs are the negative absolute values of the corresponding signal components.
The numbers on edges indicate the edge costs (most edge costs are omitted for clarity).
All capacities in the flow network are 1.
The edge costs are the vertical distances between the start and end nodes, multiplied by $\lambda$.}
\label{fig:emdflownetwork_example}
\end{figure}
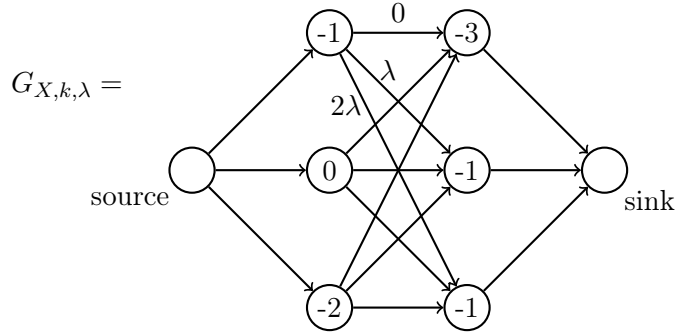

A crucial issue is the choice of the Lagrange parameter $\lambda$, which defines a trade-off between the size of the tail approximation error and the support-EMD.
Note that the optimal support $\Omega$ with parameters $k$ and $B$ does not necessarily correspond to \emph{any} setting of $\lambda$.
Nevertheless, we show that the set of supports we explore by varying $\lambda$ contains a sufficiently good approximation: the tail error and the parameters $k$ and $B$ are only increased by constant factors compared to the optimal support $\Omega$.
Moreover, we show that we can find such a good support efficiently via a binary search over $\lambda$.
Before stating our algorithm and the main result, we formalize the connection between flows and supports.

\begin{definition}[Support of a set of paths]\label{def:emdpathsupport}
Let $X \in \R^{h \times w}$ be a signal matrix, $k$ be a sparsity parameter, and $\lambda \geq 0$.
Let $P = \{q_1, \ldots, q_s\}$ be a set of disjoint paths from source to sink in $G_{X,k,\lambda}$ such that no two paths in $P$ intersect vertically (i.e., if the $q_i$ are sorted vertically and $i \leq j$, then $(u,v) \in q_i$ and $(w,v) \in q_j$ implies $u < w$).
Then the paths in $P$ define a support
\begin{equation}
\Omega_P = \{ (u,v) \, | \, (u,v) \in q_i \textnormal{ for some } i \in [s] \} \, .
\end{equation}
\end{definition}

\begin{lemma}\label{thm:emdflownetwork_correspondence}
Let $X \in \R^{h \times w}$ be a signal matrix, $k$ be a sparsity parameter and $\lambda \geq 0$.
Let $P = \{q_1, \ldots, q_s\}$ be a set of disjoint paths from source to sink in $G_{X,k,\lambda}$ such that no two paths in $P$ intersect vertically.
Finally, let $f_P$ be the flow induced in $G_{X,k,\lambda}$ by sending a single unit of flow along each path in $P$ and let $c(f_P)$ be the cost of $f_P$.
Then 
\begin{equation}
  c(f_P) = -\norm{X_{\Omega_P}}_p^p + \lambda \, \EMD(\Omega_P) \, .
\end{equation}
\end{lemma}
\begin{proof}
The theorem follows directly from the definition of $G_{X,k,\lambda}$ and $\Omega_P$.
The node costs of $P$ result in the term $-\norm{X_{\Omega_P}}_p^p$.
Since the paths in $P$ do not intersect vertically, they are a min-cost matching for the elements in $\Omega_P$.
Hence the cost of edges between columns of $X$ sums up to $\lambda \, \EMD(\Omega_P)$.
\end{proof}

For a fixed value of $\lambda$, a min-cost flow in $G_{X,k,\lambda}$ gives an optimal solution to the Lagrangian relaxation:
\begin{lemma}\label{lemma:emdflow}
Let $G_{X,k,\lambda}$ be an EMD flow network and let $f$ be an integral min-cost flow in $G_{X,k,\lambda}$.
Then $f$ can be decomposed into $s$ disjoint paths $P = \{q_1, \ldots, q_s\}$ which do not intersect vertically.
Moreover,
\begin{equation}
\label{eq:emdflowlemma}
  \norm{X - X_{\Omega_P}}_p^p + \lambda \EMD(\Omega_P) =  \min_{\Omega \in \Msupports_{k,B}} \norm{X - X_\Omega}_p^p + \lambda \EMD(\Omega) \, .
\end{equation}
\end{lemma}
\begin{proof}
Note that $\norm{X - X_\Omega}_p^p = \norm{X}_p^p - \norm{X_\Omega}_p^p$.
Since $\norm{X}_p^p$ does not depend on $\Omega$, minimizing $\norm{X - X_\Omega}_p^p + \lambda \EMD(\Omega)$ with respect to $\Omega$ is equivalent to minimizing $-\norm{X_\Omega}_p^p + \lambda \EMD(\Omega)$.

Further, all edges and nodes in $G_{X,k,\lambda}$ have capacity one, so $f$ can be composed into exactly $s$ disjoint paths $P$.
Moreover, the paths in $P$ are not intersecting vertically: if $q_i$ and $q_j$ intersect vertically, we can relax the intersection to get a set of paths $P'$ with smaller support EMD and hence a flow with smaller cost -- a contradiction.
Moreover, each support $\Omega \in \Msupports_{k,B}$ gives rise to a set of disjoint, not vertically intersecting paths $Q$ and thus also to a flow $f_Q$ with $c(f_Q) = - \norm{X_{\Omega_Q}}_p^p + \lambda \EMD(\Omega_Q)$.
Since $f$ is a min-cost flow, we have $c(f) \leq c(f_Q)$.
The statement of the theorem follows.
\end{proof}

\begin{algorithm}[!t]
\caption{Tail approximation algorithm}
\label{alg:tailapprox}
\begin{algorithmic}[1]
\Function{TailApprox}{$X,k,B,d,\delta$}
\State $x_{\min} \gets \min_{\abs{X_{i,j}} > 0} \abs{X_{i,j}}^p$
\State $\varepsilon \gets \frac{x_{\min}}{w h^2} \delta$
\State $\lambda_0 \gets \frac{x_{\min}}{2 w h^2}$
\State $\Omega \gets \textsc{MinCostFlow}(G_{X,k,\lambda_0})$
\If{$\Omega\in\Msupports_{k,B}$ and $\norm{X - X_\Omega}_p = 0$}
  \State \textbf{return} $\Omega$  \label{line:return1}
\EndIf
\State $\lambda_r \gets 0$
\State $\lambda_l \gets \norm{X}_p^p$
\While{$\lambda_l - \lambda_r > \varepsilon$} \label{line:binsearch}
  \State $\lambda_m \gets (\lambda_l + \lambda_r) / 2$
  \State $\Omega \gets \textsc{MinCostFlow}(G_{X,k,\lambda_m})$
  \If{$\EMD(\Omega) \geq B$ and $\EMD(\Omega) \leq dB$}
    \State \textbf{return} $\Omega$ \label{line:return2}
  \EndIf
  \If{$\EMD(\Omega) > B$}
    \State $\lambda_r \gets \lambda_m$
  \Else
    \State $\lambda_l \gets \lambda_m$
  \EndIf
\EndWhile
\State $\Omega \gets \textsc{MinCostFlow}(G_{X,k,\lambda_l})$
\State \textbf{return} $\Omega$ \label{line:return3}
\EndFunction
\end{algorithmic}
\end{algorithm}

We can now state our tail-approximation algorithm \textsc{TailApprox} (see Algorithm~\ref{alg:tailapprox}).
The parameters $d$ and $\delta$ for \textsc{TailApprox} quantify the acceptable tail approximation ratio (see Theorem~\ref{thm:tailapprox}).
In the algorithm, we assume that \textsc{MinCostFlow}($G_{X,k,\lambda}$) returns the support corresponding to an integral min-cost flow in $G_{X,k,\lambda}$.
Before we prove the main result (Theorem \ref{thm:tailapprox}), we show that \textsc{TailApprox} always returns an optimal result for signals $X \in \Mmodel_{k,B}$.

\begin{lemma}\label{lemma:tailapprox}
Let $x_{\min} = \min_{\abs{X_{i,j}} > 0} \abs{X_{i,j}}^p$ and $\lambda_0 = \frac{x_{\min}}{2 w h^2}$.
Moreover, let $X \in \Mmodel_{k,B}$ and $\Omega$ be the support returned by $\textsc{MinCostFlow}(G_{X,k,\lambda_0})$.
Then $\norm{X - X_\Omega}_p = 0$ and $\Omega \in \Msupports^+_{k,B}$.
\end{lemma}
\begin{proof}
Let $\Gamma = \supp(X)$, so $\Gamma \in \Msupports^+_{k,B}$.
First, we show that $\norm{X - X_\Omega}_p = 0$.
For contradiction, assume that $\norm{X - X_\Omega}_p^p > 0$, so $\norm{X - X_\Omega}_p^p \geq x_{\min} > 0$ (tail-approximation is trivial for $X=0$).
Since $\Omega$ is a min-cost flow, Lemma~\ref{lemma:emdflow} gives
\begin{align*}
  x_{\min} \; \leq \; \norm{X - X_\Omega}_p^p + \lambda_0 \EMD(\Omega) &= \min_{\Omega' \in \Msupports_{k,B}} \norm{X - X_{\Omega'}}_p^p + \lambda_0 \EMD(\Omega') \\
      &\leq 0 + \frac{x_{\min}}{2 w h^2} \EMD(\Gamma) \\
      &\leq \frac{x_{\min}}{2} \, ,
\end{align*}
which gives a contradiction.
The last line follows from $\EMD(\Gamma) \leq k h \leq n h$.

Now, we show that $\Omega \in \Msupports^+_{k,B}$.
By construction of $G_{X,k,\lambda_0}$, $\Omega$ is $s$-sparse in each column.
Moreover,
\begin{align*}
  \norm{X - X_\Omega}_p^p + \lambda_0 \EMD(\Omega) &= \min_{\Omega' \in \Msupports_{k,B}} \norm{X - X_\Omega}_p^p + \lambda_0 \EMD(\Omega') \\
     \lambda_0 \EMD(\Omega) &\leq 0 + \lambda_0 \EMD(\Gamma) \, .
\end{align*}
So $\EMD(\Omega) \leq \EMD(\Gamma) \leq B$.
\end{proof}

Next, we prove a {\em bicriterion}-approximation guarantee for \textsc{TailApprox} that allows us to use \textsc{TailApprox} as a tail approximation algorithm.
In particular, we show that one of the following two cases occurs:
\begin{description}[font=\normalfont]
\item[Case 1:] The tail-approximation error achieved by our solution is at least as good as the best tail-approximation error achievable with support-EMD $B$.
The support-EMD of our solution is at most a constant times larger than $B$.
\item[Case 2:] Our solution has bounded tail-approximation error and support-EMD at most $B$.
\end{description}

In order to simplify the proof of the main theorem, we use the following shorthands:
$\Omega_l = \textsc{MinCostFlow}(G_{X,k,\lambda_l})$,
$\Omega_r = \textsc{MinCostFlow}(G_{X,k,\lambda_r})$,
$b_l = \EMD(\Omega_l)$, 
$b_r = \EMD(\Omega_r)$, 
$t_l = \norm{X - X_{\Omega_l}}_p^p$, and $t_r = \norm{X - X_{\Omega_r}}_p^p$.

\begin{restatable}{theorem}{tailapprox}\label{thm:tailapprox}
Let $d > 1$, $\delta > 0$, and let $\Omega$ be the support returned by \textsc{TailApprox}($X,k,B,d,\delta$).
Let $\OPT$ be the tail approximation error of the best support with support-EMD at most $B$, i.e., $\OPT = \min_{\Gamma\in\Msupports_{k,B}} \norm{X - X_\Gamma}_p^p$.
Then at least one of the following two guarantees holds for $\Omega$:
\begin{description}
\item[Case 1:] $B \leq \EMD(\Omega) \leq d B$ and $\norm{X - X_\Omega}_p^p \leq \OPT$
\item[Case 2:] $\EMD(\Omega) \leq B$ and $\norm{X - X_\Omega}_p^p \leq (1 + \frac{1}{d-1} + \delta) \OPT$.
\end{description}
\end{restatable}

\begin{proof}
We consider the three cases in which \textsc{TailApprox} returns a support.
If \textsc{TailApprox} returns in line~\ref{line:return1}, the first guarantee in the theorem is satisfied.
If \textsc{TailApprox} reaches the binary search (line \ref{line:binsearch}), we have $X \notin \Mmodel_{k,B}$ (the contrapositive of Lemma \ref{lemma:tailapprox}).
Therefore, we have $\OPT \geq x_{\min} > 0$ in the remaining two cases.

If \textsc{TailApprox} returns in line~\ref{line:return2}, we have $B \leq \EMD(\Omega) \leq d B$.
Moreover, Lemma~\ref{lemma:emdflow} gives
\begin{align*}
  \norm{X - X_\Omega}_p^p + \lambda_m \EMD(\Omega) &\leq \min_{\Omega' \in \Msupports_{k,B}} \norm{X - X_{\Omega'}}^p_p + \lambda_m \EMD(\Omega') \\
      &\leq \OPT + \lambda_m B \, .
\end{align*}
Since $\EMD(\Omega) \geq B$, we have $\norm{X - X_\Omega}_p^p \leq \OPT$.

We now consider the third return statement (line \ref{line:return3}), in which case the binary search terminated with $\lambda_l - \lambda_r \leq \varepsilon$.
In the binary search, we maintain the invariant that $b_l \leq  B$ and $b_r > dB$.
Note that this is true before the first iteration of the binary search due to our initial choices of $\lambda_r$ and $\lambda_l$.\footnote{Intuitively, our initial choices make the support-EMD very cheap and very expensive compared to the tail approximation error.}
Moreover, our update rule maintains the invariant.

We now prove the bound on $\norm{X - X_\Omega}_p^p = t_l$.
From Lemma~\ref{lemma:emdflow} we have
\begin{align*}
  t_r + \lambda_r b_r &\leq \OPT + \lambda_r B \\
  \lambda_r dB &\leq \OPT + \lambda_r B \\
  \lambda_r &\leq \frac{\OPT}{B(d-1)} \, .
\end{align*}
Since the binary search terminated, we have $\lambda_l \leq \lambda_r + \varepsilon$.
We now combine this inequality with our new bound on $\lambda_r$ and use it in the following inequality (also from Lemma~\ref{lemma:emdflow}):
\begin{align*}
  t_l + \lambda_l b_l &\leq \OPT + \lambda_l B \\
                  t_l &\leq \OPT + \lambda_l B \\
                      &\leq \OPT + (\lambda_r + \varepsilon) B \\
                      &\leq \OPT + \frac{\OPT}{d-1} + \varepsilon B \\
                      &\leq \parens{1 + \frac{1}{d-1}} \OPT + \frac{x_{\min} \delta B}{w h^2} \\
                      &\leq \parens{1 + \frac{1}{d-1}} \OPT + \delta x_{\min} \\
                      &\leq \parens{1 + \frac{1}{d-1} + \delta} \OPT \, .
\end{align*}
This shows that the second guarantee of the theorem is satisfied.
\end{proof}

\begin{restatable}{corollary}{tailapproxcorollary}\label{corollary:tailapprox}
Let $p \geq 1$, $c > 1$, $0 < \delta < c - 1$, and $d = 1 + \frac{1}{c - \delta - 1}$.
Then \textsc{TailApprox} is a $(c^{1/p}, \Msupports_{k,B}, \Msupports_{k, dB}, p)$-tail approximation algorithm.
\end{restatable}
\begin{proof}
The tail approximation guarantee follows directly from Theorem~\ref{thm:tailapprox}.
Note that we cannot control which of the two guarantees the algorithm returns.
However, in any case we have $\EMD(\Omega) \leq dB$, so $\Omega \in \Msupports_{k,dB}$.
\end{proof}

In order to simplify the time complexity of \textsc{TailApprox}, we assume that $h = \Omega(\log w)$, i.e., the matrix $X$ is not very ``wide'' and ``short''.
We arrive at the following result.
\begin{theorem}\label{thm:tailapproxtime}
Let $\delta > 0$, $x_{\min} = \min_{\abs{X_{i,j}} > 0} \abs{X_{i,j}}^p$, and $x_{\max} = \max \abs{X_{i,j}}^p$.
Then \textsc{TailApprox} runs in $O(s  n  h (\log\frac{n}{\delta} + \log\frac{x_{\max}}{x_{\min}}))$ time.
\end{theorem}
\begin{proof}
We can solve our instances of the min-cost flow problem by finding $s$ augmenting paths because all edges and nodes have unit capacity.
Moreover, $G_{X,k,\lambda}$ is a directed acyclic graph, so we can compute the initial node potentials in linear time.
Each augmenting path can then be found with a single run of Dijkstra's algorithm, which costs $O(wh \log(wh) + w h^2) = O(nh)$ time \cite{CLRS01}.
The number of iterations of the binary search is at most
\begin{equation}
\log\frac{\norm{X}_p^p}{\epsilon} \;\; = \;\; \log\frac{\norm{X}_p^p n h}{x_{\min} \delta} \;\; \leq \;\; \log\frac{x_{\max} n^2 h}{x_{\min} \delta} \;\; \leq \;\; \log\frac{n^3}{\delta} + \log\frac{x_{\max}}{x_{\min}} \; .
\end{equation}
Combining this with a per-iteration cost of $O(snh)$ gives the stated running time.
\end{proof}

To summarize, the algorithm proposed in~\cite{HIS} satisfies the criteria of a tail-approximation oracle.
This, in conjunction with the head approximation oracle proposed in Section~\ref{sec:head}, gives a full sparse recovery scheme for the CEMD model, which we describe below.

\subsection{Compressive Sensing Recovery}
We now bring the results from the previous sections
together. Specifically, we show that AM-IHT
(Algorithm~\ref{alg:approxmodeliht}), equipped with
\textsc{HeadApprox} and \textsc{TailApprox},
constitutes a model-based compressive sensing recovery algorithm
that significantly reduces the number of measurements necessary for
recovering signals in the CEMD model. The main result is the following
theoretical guarantee:
\begin{restatable}{theorem}{recoverythm}\label{thm:recovery}
Let $x\in\Mmodel_{k,B}$ be an arbitrary signal in the CEMD model with dimension $n = w h$.
Let $A\in\R^{m\times n}$ be a measurement matrix with i.i.d.\ Gaussian entries and let $y \in \R^m$ be a noisy measurement vector, i.e., $y = A x + e$ with arbitrary $e \in \R^m$.
Then we can recover a signal approximation $\widehat{x} \in \Mmodel_{k,2B}$ satisfying
$\norm{x - \widehat{x}}_2 \leq C \norm{e}_2$
for some constant $C$ from $m = O(k \log(\frac{B}{k} \log\frac{k}{w}))$ measurements.
Moreover, the recovery algorithm runs in time $O(n \log\frac{\norm{x}_2}{\norm{e}_2} (k \log n + \frac{k h}{w} (B + \log n + \log\frac{x_{\max}}{x_{\min}})))$ where $x_{\min} = \min_{\abs{x_i} > 0} \abs{x_i}$ and $x_{\max} = \max \abs{x_i}$.
\end{restatable}
\begin{proof}
First, we show that $m$ rows suffice for $A$ to have the desired model-RIP.
Following the conditions in Corollary \ref{cor:boosting}, $A$ must satisfy the $(\delta, \Msupports_{k,B} \supportplus \Msupports_T \supportplus \Msupports_H^{\supportplus t})$-model-RIP for small $\delta$, where $t$ is the number of times we boost \textsc{HeadApprox} (a constant depending on $\delta$ and $c_T$).
We have $\Msupports_T = \Msupports_{k,2B}$ from Corollary \ref{corollary:tailapprox} and $\Msupports_H = \Msupports_{2k, 3 \gamma B}$ where $\gamma = \ceil{\log\frac{k}{w}} + 1$ from Theorems \ref{thm:cemdaddition} and \ref{thm:headapproxbasic} (note that \textsc{HeadApprox} must be a $(c_H, \Msupports \supportplus \Msupports_T, \Msupports_H, 2)$-head-approximation oracle).
Invoking Theorem \ref{thm:cemdaddition} again shows that it suffices for $A$ to have the $(\delta, \Msupports_{(2 + 2t)k, (3 + 3 t \gamma) B})$-model-RIP.
Using Theorem \ref{thm:samplingbound} and the fact that $t$ is a constant, Fact \ref{fact:modelripbound} then shows that
\[
m \; = \; O\left(k \log\frac{\gamma B}{k} \right) \; = \; O \left(k \log\left(\frac{B}{k} \log\frac{k}{w}\right)\right)
\]
suffices for $A$ to have the desired model-RIP.

Equipped with our model-RIP, we are now able to invoke Corollary \ref{cor:boosting}, which directly gives the desired recovery guarantee $\norm{x - \widehat{x}}_2 \leq C \norm{e}_2$.
Moreover, the corollary also shows that the number of iterations of AM-IHT is bounded by $O(\log \frac{\norm{x}_2}{\norm{e}_2})$.
In order to prove our desired time complexity, we now only have to bound the per-iteration cost of AM-IHT.

In each iteration of AM-IHT, the following operations have a relevant time complexity:
(i) Multiplication with $A$ and $A^T$.
The measurement matrix has at most $k \log n$ rows, so we bound this time complexity by $O(n k \log n)$.
(ii) \textsc{HeadApprox}.
From Theorem \ref{thm:headapproxtime} we know that \textsc{HeadApprox} runs in time $O(n \frac{kh}{w} B)$.
(iii) \textsc{TailApprox}.
Theorem \ref{thm:tailapproxtime} shows that the tail-approximation algorithm runs in time $O(n \frac{kh}{w} (\log n + \log \frac{x_{\max}}{x_{\min}}))$.
Combining these three bounds gives the running time stated in the theorem.
\end{proof}

Note that for $B = O(k)$, the measurement bound gives $m = O(k
\log\log \frac{k}{w})$, which is a significant improvement over the
standard compressive sensing measurement bound $m = O(k \log
\frac{n}{k})$. In fact, the bound for $m$ is only a $\log \log \frac{k}{w}$ factor away from
the information-theoretically optimal bound $m = O(k)$. We leave it as
an open problem whether this spurious factor can be eliminated via a
more refined analysis or algorithm.

\section{Conclusions}
\label{sec:conc}

We have introduced a new framework called \emph{approximation-tolerant model-based compressive sensing}. Our framework consists of a range of algorithms for model-based compressive sensing that succeed even when the model-projection oracles are approximate.  All our algorithms involve oracles that provide constant-factor approximations to both the ``head'' and ``tail'' versions of the model-projection problem. We have instantiated these algorithms for the Constrained Earth Mover Distance (CEMD) model. To achieve this, we have designed novel polynomial-time head- and tail-approximation oracles for the CEMD model based on graph optimization techniques. Leveraging these oracles and our framework results in nearly sample-optimal recovery schemes for signals belonging to this model.

%
Several avenues for future work remain. We have developed model-based recovery schemes that succeed with dense measurement matrices (AM-IHT, AM-CoSaMP), as well as sparse matrices (AM-IHT with RIP-1). An interesting question is whether model-based recovery can be extended to other classes of measurement matrices, such as subsampled Fourier matrices~\cite{RV08}. Also, the required sample-complexity $m$ specified by Theorem~\ref{thm:recovery} is a factor of $ \log(\frac{B}{k} \log\frac{k}{w})$ away from the optimal $m = O(k)$, and it is possible that a different approach is needed to remove this log-factor. Finally, finding an efficient algorithm (or proving a computational hardness result) for \emph{exact} projections into the CEMD model remains an open question.

\bibliographystyle{IEEEbib}
\bibliography{main,csbib}

\end{document}